\documentclass[a4paper]{article}
\usepackage[a4paper,top=3cm,bottom=2cm,left=3cm,right=3cm,marginparwidth=1.75cm]{geometry}

\usepackage[english]{babel}
\usepackage[utf8x]{inputenc}
\usepackage[T1]{fontenc}
\usepackage{soul}

\usepackage{lipsum}
\let\OLDthebibliography\thebibliography
\renewcommand\thebibliography[1]{
  \OLDthebibliography{#1}
  \setlength{\parskip}{0pt}
  \setlength{\itemsep}{0pt plus 0.3ex}
}

\usepackage{amsthm}
\usepackage{amsmath}
\usepackage{amssymb}
\usepackage{tikz-cd}
\usepackage{multirow}
\usepackage{comment}

\newtheorem{theorem}{Theorem}[section]
\newtheorem{proposition}[theorem]{Proposition}
\newtheorem{lemma}[theorem]{Lemma}
\newtheorem{corollary}[theorem]{Corollary}
\theoremstyle{definition}
\newtheorem{definition}[theorem]{Definition}
\newtheorem{example}[theorem]{Example}
\newtheorem{remark}[theorem]{Remark}
\newtheorem{notation}[theorem]{Notation}

 \usepackage[normalem]{ulem}
 \usepackage{colortbl,xcolor}

\title{On the  stability of persistent entropy and new summary functions  for Topological Data Analysis}
\author{N. Atienza, R. Gonzalez Diaz, M. Soriano Trigueros}

\def\then{\Rightarrow}
\def\nn{\nonumber}
\def\RR{\mathbb{R}}
\def\ee{\varepsilon}
\def\sumri{\sum_{i=1}^{n_{r}}}
\def\sumi{\sum_{i=1}^{n_{p}}}
\def\suma{\sum_{i=1}^{n_{a}}}

\def\zabri{w_{r,i}}
\def\zabpi{w_{p,i}}
\def\zabqi{w_{q,i}}
\def\Im{\mbox{Im}\,}

\title{On the  stability of persistent entropy and new summary functions  for Topological Data Analysis}
\date{}

\begin{document}
	\maketitle
		\begin{center}
		\centering Department of Applied Mathematics I, University of Seville\\
		\{natienza, rogodi, msoriano4\}@us.es
		\end{center}
	\begin{abstract}
		Persistent homology and persistent entropy have  recently become useful tools for patter recognition. In this paper, we find requirements under which persistent entropy is stable to small perturbations in the input data and scale invariant. In addition, we  describe two new stable summary functions combining persistent entropy and the Betti curve. Finally, we use the previously defined summary functions  in a material classification task to show their usefulness in machine learning and pattern recognition.
	\end{abstract}
	
	\section{Introduction}

		Topological data analysis (TDA) uses computational topology tools to study datasets. Intuitively, topological features like homology can be seen as qualitative geometric properties related to the notions of proximity and continuity and, therefore, can be useful tools for pattern recognition \cite{mrozek}.  
		TDA has become a large field of research, with persistent homology (and its precursor known as  size functions \cite{sizefun}) as its key tool. It has 
	 been applied successfully  in many areas (see, for example, \cite{ferry}).
		Its standard workflow is the following (see also   Figure~\ref{figure:1}): 
		\begin{enumerate}
			\item Start with a dataset, for example, a point cloud, endowed with some notion of proximity (usually a metric).
			\item Depending on the kind of information we want to obtain, build a simplicial complex and a filter function on it. Compute a nested sequence of increasing subcomplexes (which encapsulate features from data) using the filter function.
			\item Compute the homology of each subcomplex (intuitively, homology captures the ``holes'' of the underlying space) and study how it evolves in the sequence, leading to the key concept of persistent homology.
		\end{enumerate}

Persistent homology can be compactly represented using  persistence barcodes \cite{barcodes},  diagrams \cite{diagrams}
		and, more recently, landscapes \cite{Landscape}. There exist stability results showing that these representations are robust under small perturbations of the given data (see, for example, \cite{Computational}).
In addition, there are numerous software packages to calculate  persistent homology and its representations.
		A nice study of the performance of  available software packages  is made in \cite{Review}.

		\begin{figure}[ht]
			\centering
			\[
				\begin{tikzcd}[column sep=tiny]
				\includegraphics[scale=0.2]{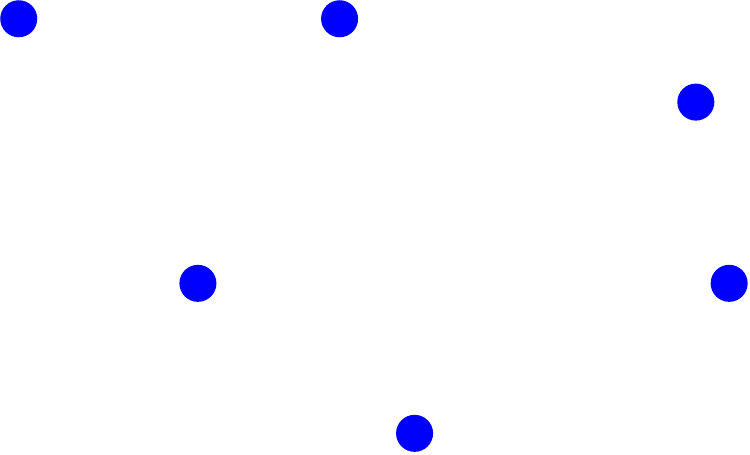}\arrow{dr}&&
				\includegraphics[scale=0.35]{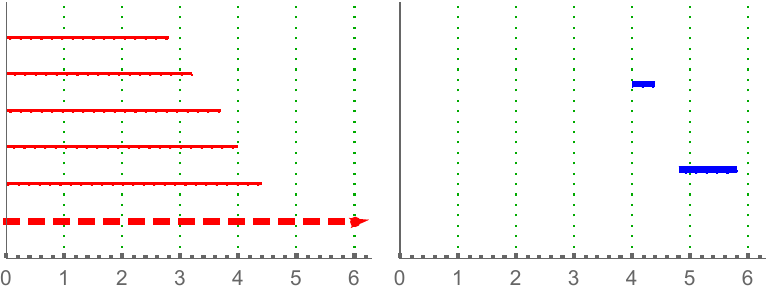} \\		&\includegraphics[scale =0.3]{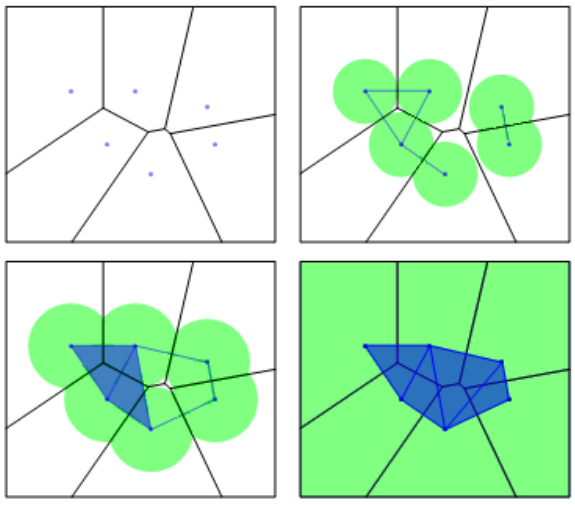} \arrow{ur}&     
			\end{tikzcd}
			\]
			\caption{Standard workflow in topological data analysis.}
			\label{figure:1}
		\end{figure}
		
Although persistence barcodes,  diagrams, and landscapes are metric spaces 
used to compare persistent homology of  datasets, persistence barcodes and diagrams do not work properly for statistical analysis. For example, they fail to have unique mean (see \cite{wa-sta}). Persistence landscapes perform better  \cite{Landscape}, but they are limited to the context of probability in Banach spaces. 
It is more useful sometimes to summarize the information contained in persistent homology using only a number. 
It becomes especially appropriate when only small samples are available since univariate non-parametric tests are required in these cases.
Persistent entropy seems to be a perfect candidate to summarize  persistent homology using only a number. Specifically, persistent entropy is the Shannon entropy \cite{Shannon} of a probability distribution obtained from persistent homology.
It was defined in its current form in \cite{Persistent-entropy} but a precursor of this definition appears in \cite{Barcodes-entropy}. Some successful applications of persistent entropy have been developed for pattern recognition of signals \cite{Eliptic,Piecewise}, complex systems \cite{Complex-system}, biological images \cite{cells} and clustering \cite{cluster}. A more theoretical approach allows   persistent entropy to be used to distinguish topological features from noise \cite{Noise}.
With regards to implementation, persistent entropy has already been implemented as a method in Gudhi library\footnote{
   https://github.com/GUDHI/gudhi-devel/blob/master/src/python/gudhi/\\representations/vector\_methods.py },  scikit-TDA library\footnote{https://github.com/scikit-tda/persim} and 
   Giotto library\footnote{https://github.com/giotto-ai/giotto-learn/blob/master/giotto/diagrams/features.py}.
Some partial results about stability of persistent entropy have been given in \cite{Piecewise,Noise} but, as far as we know,  no formal study of persistent entropy has been done.
The main objective of this paper is to provide a general stability result for persistent entropy and to study under which conditions persistent entropy is scale-invariant.
		
When it is not necessary 		 to find significant differences in data but a classification task is needed, the usual approach is to replace 		statistical tests with machine learning methods. In this case, 		summarizing persistent homology in a number may be too restrictive, since we are projecting an infinite-dimensional space (persistence barcodes) to only one dimension (persistent entropy). One solution might be to use summary functions instead.
Common approaches to
summarizing persistence barcodes  include kernel functions such as 
persistence multi-scale  kernel \cite{multi-scale},
persistence weighted Gaussian kernel
\cite{32} and sliced Wasserstein kernel \cite{Oudot}, as well as persistence-diagram vectorizations such as the already mentioned persistence landscape, persistence silhouettes \cite{Silhouettes}, persistence images \cite{adams}, Euler characteristic curves
\cite{euler}, 
topological intensity maps 
\cite{cosmic} and Betti curves \cite{2019,Umeda_2017}.
In this paper, we will define two new stable summary functions based on persistent entropy that can be used as a complementary function to the previous ones to describe persistence barcodes.
		
The paper is organized as follows. After recalling the theory of persistent homology in Section 2, stability and  scale-invariance of persistent entropy is introduced in Section 3.
In Section 4, we define two new summary functions derived from the concept of persistent entropy and study also their stability. Examples showing the applicability of these functions are also given. The usefulness of the summary functions defined in this paper is showed in Section 5. The paper ends with a section devoted to conclusions and future work.


\section{Background}
		
In this section, we give a quick overview of how algebraic topology is applied to data analysis.
An instructive book showing the main algebraic topology tools for data analysis is \cite{Computational}. 
		
As explained in the introduction, in order 	to apply algebraic topology tools to data analysis, 	we first must summarize the information provided by the data in a combinatorial structure, the simplicial complex structure being the most commonly used. 
Recall that an $n$-simplex is the convex hull of $n+1$ affinely independent points.
A $0$-simplex is a point, a $1$-simplex is a segment, a $2$-simplex is a triangle, a $3$-simplex is a tetrahedron and so on. A simplicial complex is a set of simplices glued in a specific way.
An abstract simplicial complex can be seen as a way of storing  the  combinatorial structure of a simplicial complex.
		
		\begin{definition}[abstract 
		simplicial complex]
			Let $X$ be a finite set. A family $K$ of subsets  of $X$ is an abstract \emph{simplicial complex} if for every subsets $\sigma \in K$ and $\sigma' \subseteq X$, we have that $\sigma' \subset \sigma$ implies $\sigma' \in K$ (in other words, non-empty intersections of simplices in $K$ are also simplices of $K$). A subset in $K$ of $m+1$ elements of $X$ is called an $m$-simplex. 
		\end{definition}
When the finite set $X$ represents data, the geometrical structure of its associated simplicial complex can provide information about how the data is related. Usually, these relations are not equally significant 	so it is common to define an order in its simplices to represent their importance. This can be done implicitly using a filter function.
		
		\begin{definition}[filtration]
			A  \emph{filter function} on a simplicial complex $K$ is a monotonic function $f : K \rightarrow \RR$
			satisfying that
			$\sigma' \subset \sigma$ implies $f(\sigma') \leq f(\sigma)$. A \emph{filtration} on $K$, obtained from $f$, is the sequence of subcomplexes $\big( K_t\big)_{t\in \RR}$ where $K_t = f^{-1}(-\infty,t]$. 
		\end{definition}
		Notice that, because of the monotonicity of $f$, 
		the set $K_t$ is a simplicial complex for all $t$, and $t_1 < t_2$ implies that  $K_{t_1} \subseteq K_{t_2}$. To help intuition, the parameter $t$ will be referred as {\it time} although its physical meaning may be completely different. The following definition is an example of filtration and requires $X$ to be a metric space.
		
		\begin{definition}[Vietoris-Rips filtration]
			Let $X$ be a 
			finite set of points  endowed with a distance $d_X$.
			The Vietoris-Rips filtration of $X$
			is the sequence $\big( Rips(X,t)\big)_{t\in \RR}$ obtained from the filter function 
			$$f([x_0,\ldots,x_m]) = \max_{0\leq i,j\leq m} d_X(x_i,x_j)$$
			where, for each $t\in \RR$, the simplices of the Vietoris-Rips simplicial complex $Rips(X,t)$ are defined as:
			\begin{equation*}
			\sigma = \langle x_0, \ldots, x_m\rangle\in Rips(X,t) \Longleftrightarrow  f([x_0,\ldots,x_m]) \leq t. 
			\end{equation*}
		\end{definition}

		Homology groups of simplicial complexes provide a formal interpretation of what an $n$-dimensional ``hole'' is. Intuitively, 
		a $0$-dimensional hole is a  connected component,
		a $1$-dimensional hole  is a loop, a $2$-dimensional hole is a cavity, and so on.
		Given a simplicial complex $K$, an $m$-chain $c$ is a formal sum of $m$-simplices of $K$. 
		That is,
		$c = \sum_{i=1}^k a_i \sigma_i$ where, for $1\leq i\leq k$, $\sigma_i$ is an $m$-simplex of $K$ and 
		$a_i$ is a  coefficient in an unital ring $R$.
		To relate the $m$-chains of a given simplicial complex $K$ with its $m$-dimensional holes, we need the boundary operator $\partial_m$:
		If $\langle x_0,\ldots,x_m\rangle$ is an $m$-simplex of $K$ then,
		\begin{equation*}
		\partial_m (\langle x_0,\ldots,x_m\rangle) = \sum_{i=0}^m(-1)^i \langle x_0,\ldots,x_{i-1},x_{i+1},\ldots,x_{m}\rangle.
		\end{equation*}
		We can extend this definition to any $m$-chain by linearity. Notice that $\partial_{m-1} \circ \partial_{m} = 0$ or, in other words, the boundary of a boundary is null. 
		The $m$-dimensional holes of $K$ are detected from $m$-chains whose boundary is zero without being  ``boundaries'' themselves. 
		More concretely, the $m$-dimensional homology group of $K$ is defined as the quotient group
	\begin{equation*}
		H_m(K) = \dfrac{\emph{Ker } \partial_m}{\emph{Im } \partial_{m+1}},
		\end{equation*}
		and its $m$-dimensional Betti number as
		$\beta_m = \emph{rank } H_m(K)$. 
		Intuitively, $\beta_0$ 
		counts the number of independent connected components of $K$, $\beta_1$ the number of independent loops, and so on. 
		
		When computed over a field, the homology groups is actually a vector space. This fact allow us tu use persistent homology to study filtrations. 
		
		\begin{definition}[persistent homology]
		Let $\mathcal{F}=\big(K_t\big)_{t\in\RR}$ be a filtration.
		Suppose the ground ring $R$ is a field and, therefore,  for each $t\in\RR$ and $m\in \mathbb{Z}$,
	the $m$-dimensional homology group $H_m(K_t)$  is a vector space.
		For every $a< b$ and $m$, consider the linear maps $v_m^{a,b}:H_m(K_a)\rightarrow H_m(K_b)$ induced by the inclusion $K_a \hookrightarrow K_b$. 
		The $m$-th \emph{persistent homology groups} are the images of the linear maps $v_m^{a,b}$, denoted by 
		$\Im v_m^{a,b}$. The set $\{\Im v_m^{a,b}\}_{a<b}$
			is called the $m$-th \emph{persistent homology} of the  filtration $\mathcal{F}$ and is denoted by $\mathcal{H}_m$.
		\end{definition}
We  assume that the rank of $H_m(K_t)$ is finite for all $t\in \RR$ and $m\in \mathbb{Z}$. In this case, persistent homology can be compactly represented via persistence barcodes (or diagrams). 
		
			\begin{figure}[t!]
				\begin{center}
					\begin{tabular}{ |c||c| }
						\hline\multicolumn{2}{|c|}
						{\hspace{0.1cm}
							\includegraphics [width=0.7\textwidth] {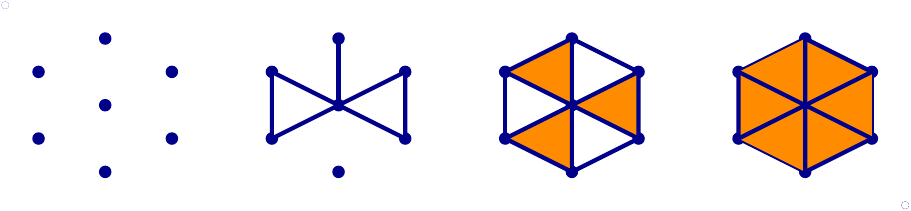}
							\hspace{0.1cm}}
						\\
						\hline
						\hline	
						\includegraphics [width=0.35\textwidth]{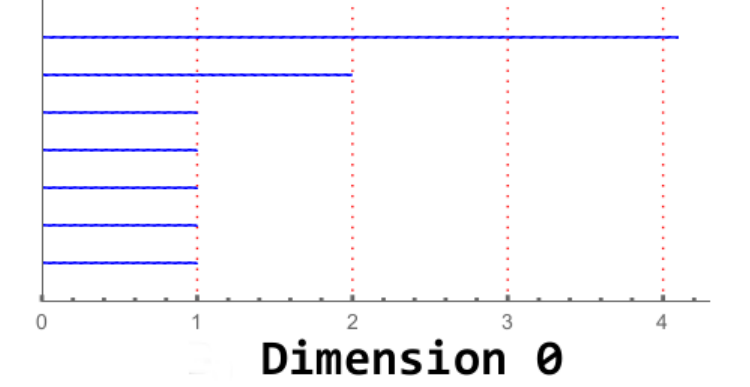}
						&
						\includegraphics [width=0.35\textwidth]{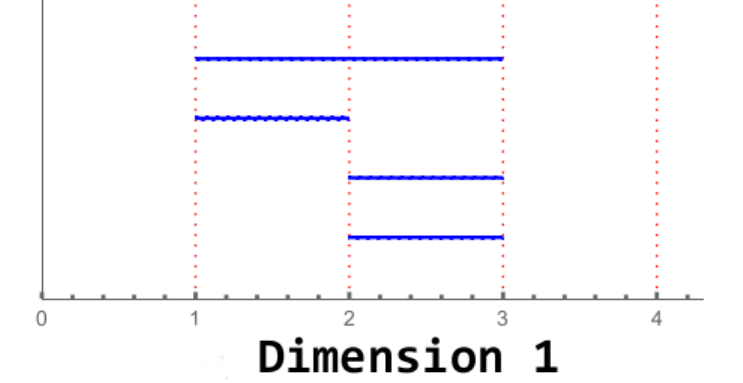}
						\\
						\hline
					\end{tabular}
				\end{center}
				\caption{ Top: example of a filtration $\mathcal{F}$. Bottom: $0$-th and $1$-st persistence  barcodes of $\mathcal{F}$.}\label{figure:fig_bar}
			\end{figure}

		\begin{definition}[persistence  barcodes] 
	Let	 $\mathcal{H}_m$ be the $m$-th persistent homology of a filtration $\mathcal{F}$.
		For $a<b$ and $m\in\mathbb{Z}$, define $\mu_m^{a,b} = \big(\emph{rank }(\Im v_m^{a,b-1}) - \emph{rank }(\Im v_m^{a,b})\big)$ $-\big(\emph{rank }(\Im v_m^{a-1,b-1}) - \emph{rank }(\Im v_m^{a-1,b})\big)$ that can be interpreted as the number of $m$-dimensional homology classes which are ``born'' at time $a$ and ``die'' at time $b$.
		Then,  $\mathcal{H}_m$	can be represented by  the multiset\footnote{A multiset is a set whose elements can be repeated} of intervals $\big\{[x_i,y_i)\big\}_{1\leq i\leq n}$,
		called the $m$-th \emph{persistence barcode} or diagram of $\mathcal{H}_m$,
		where each interval $[x_i,y_i)$ appears $\mu_m^{x_i, y_i}$ times.
		\end{definition}
In this paper, we assume barcodes have a finite number of elements, find an example in Figure~\ref{figure:fig_bar}. We introduce now the notation  used along the paper.
		\begin{notation}\label{notation}
		Let $\mathcal{B}$ denote the set of persistence barcodes. 
		Given a persistence barcode $A\in \mathcal{B}$, 
		its $n_a$ intervals will be denoted by $[x^a_i,y^a_i)$ for $1\leq i\leq n_a$. Besides, the length of  $[x^a_i,y^a_i)$ will be denoted by $\ell^a _i$, that is, $\ell^a_i=y^a_i-x^a_i$. Finally, $L_a$ will denote the sum $\suma\ell^a_i$.
		Moreover, given two persistence barcodes 
	$A$ and $B$, denote $\max\{n_a,n_b\}$ by $n_{max}$ and
 $\max\{L_a,L_b\}$ by 
	 $L_{\max}$. 
		\end{notation}
Let us define the following  subsets of  $\mathcal{B}$. 
		\begin{definition}
		The set of finite persistence barcodes is defined as:
		\[
		\mathcal{B}_F = \{ A\in \mathcal{B}\;\mbox{ such that } \; y^a_i < \infty \,\mbox{ for all }[x_i^a,y_i^a) \in A \}.
	\]
	The set of persistence barcodes whose intervals all start at $0$ is denoted as $\mathcal{B}_0$, that is: 
	\[
		\mathcal{B}_0 = \{ A\in \mathcal{B} \;\mbox{ such that } \; x^a_i = 0 \,\mbox{ for all } [x^a_i,y^a_i) \in A\}
	\]
	And, finally,  the set of normalized persistence barcodes is defined as:
	\[
		\mathcal{B}_N = \Big\{ A\in \mathcal{B} \;\mbox{ such that } \; \suma \ell^a_i = 1\Big \}.
	\]
	\end{definition}
In the sequel, we will assume that $n_{a}>1$ for all $A\in \mathcal{B}_F$ to avoid degenerate cases. There is a  correspondence  between persistence barcodes in $\mathcal{B}_F$  and persistence barcodes in $\mathcal{B}_0 \cap \mathcal{B}_N$. 
	\begin{definition}
	Let $\psi: \mathcal{B}_F \rightarrow \mathcal{B}_0 \cap 			\mathcal{B}_N $  be the projection defined as the  composition: $\psi = \phi \circ \pi$  where $\phi$ and $\pi$ are defined as follows
(see Figure~\ref{figure:fig_comp}):
	\begin{align*}
		& \phi: \mathcal{B}_F \rightarrow \mathcal{B}_N 
		\;\mbox{ where }     
	A=\big\{[x_i^a,y_i^a) \big\}_{1\leq i\leq n_a} \mapsto \phi(A)=\left\{\left[ \frac{x^a_i}{L_a},\frac{y^a_i}{L_a}\right) \right\}_{1\leq i\leq n_a}
	\\
		&  \pi : \mathcal{B}_F \rightarrow \mathcal{B}_0  \;\mbox{ where }
		A=\big\{ [x^a_i,y^a_i) \big\}_{1\leq i\leq n_a} \mapsto  \pi(A)= \big\{[0, \ell_i^a) \big\}_{1\leq i\leq n_a}
	\end{align*}
	\end{definition}

			\begin{figure}[ht]
		\begin{center}
			\begin{tikzcd}[column sep=tiny, row sep=tiny]
				\includegraphics[width=0.29\textwidth]{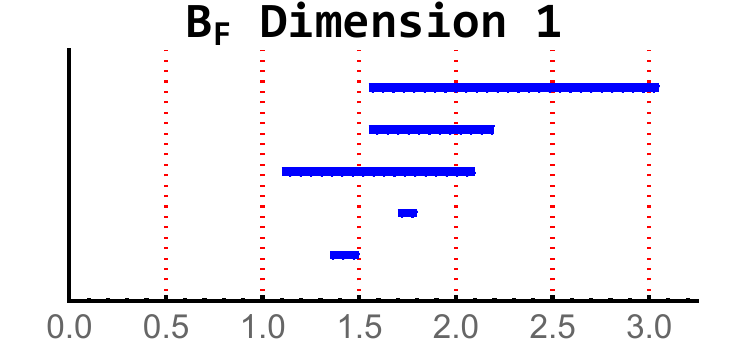}\arrow[ bend right=-25]{rr}{\boldsymbol{\psi}}\arrow[ bend right=10]{dr}{\boldsymbol{\pi}}&&
				\includegraphics[width=0.29\textwidth]{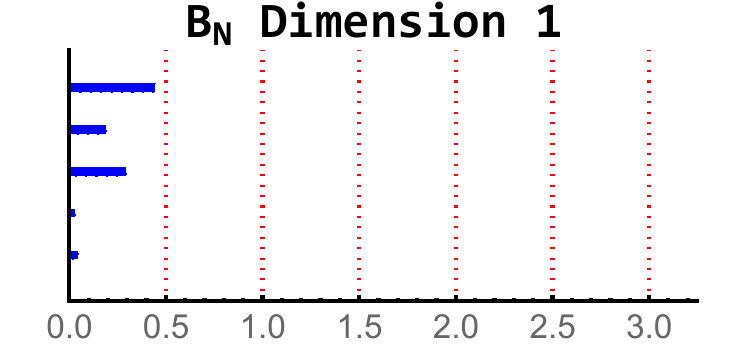}\\			&\includegraphics[width=0.29\textwidth]{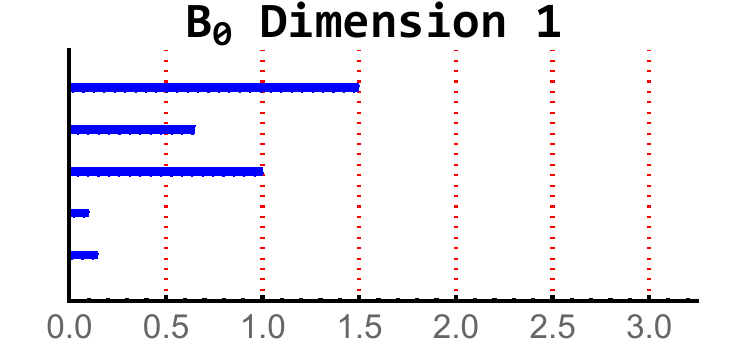}\arrow[ bend right=10]{ur}{\boldsymbol{\phi}}&     
			\end{tikzcd}
		\end{center}
		\caption{Example of projections $\pi$, $\phi$ and $\psi$.}\label{figure:fig_comp}
	\end{figure}
	
		The following 
		metrics can be defined on $\mathcal{B}$.

		\begin{definition}[Wasserstein and bottleneck distances]\label{Was-definition}
			Let $A,B \in \mathcal{B}$ and $1\leq p < \infty$. 
			Define the \emph{$p$-th Wasserstein  distance} as
			\begin{equation*}
			d_p(A,B) = \left( \min_{\gamma}\sum_{i=1}^{n_{\gamma}} \max \big\{ |x_i^a - x_{\gamma(i)}^b|^p, |y_i^a - y_{\gamma(i)}^b| ^p\big\}\right)^{\frac{1}{p}}
			\end{equation*}
			where  $\gamma$ is any bijection between the multisets 
			 $A=\{[x_i^a,y_i^a)\}_{1\leq i\leq n_a}$ and $B=\{[x_i^b,y_i^b)\}_{1\leq i\leq n_b}$ (including, if necessary, intervals $[t,t)$ of zero length) and $n_{\gamma}$ is the cardinality\footnote{Since $\gamma$ is a bijection, cardinality of $\gamma$ refers to the number of elements of the domain of $\gamma$ which coincides with the number of elements of the image of $\gamma$.} of $\gamma$. 
			\\
			The limit case $p=\infty$ is called the {\it bottleneck distance} and is defined by  
			\begin{equation*}
			d_\infty(A,B) =  \min_{\gamma}\max_{i 
			}  \max \big\{ |x_i^a - x_{\gamma(i)}^b|, |y_i^a - y_{\gamma(i)}^b| \big\}  .
			\end{equation*}
		\end{definition}
		
Observe that $n_{\max}\leq n_{\gamma}\leq n_a+n_b$.
Besides, in case $y_i^a$ or $y_{\gamma(i)}^b$ is $\infty$ then $|y_i^a - y_{\gamma(i)}^b|$ is set to $\infty$. 
In case both $y_i^a$ and $y_{\gamma(i)}^b$ are $\infty$ then $|y_i^a - y_{\gamma(i)}^b|$  is set to $0$.
Notice also that we have replaced the $\inf$ and $\sup$ terms of  the original definition of Wasserstein and bottleneck distance \cite[p.~180-183]{Computational} by  $\min$ and $\max$ terms because, in this paper,  persistence barcodes have always a finite number of intervals.

We finish this section with some well-known persistent homology stability results, supporting the idea that an algorithm designed using persistent homology tools will produce ``similar'' outputs for ``similar'' inputs. 
		
		\begin{theorem}[\cite{Wstability}]\label{Est1}
			Let $f,g\, : \, X \rightarrow \RR$ be two tame\footnote{The function $f$ is \emph{tame} if 
there is a finite number of different elements 
in the set
$\{H_m(f^{−1}(−\infty, a]))\}_a$
and such set consists of homology groups
whose ranks are finite.} Lipschitz functions on a metric space $X$ whose triangulations grow polynomially with constant exponent $j\geq 1$. Then, there are constants $c\geq 1$ and $k\geq j$ such that the $p$-th Wasserstein distance between their corresponding 	persistence barcodes, denoted by  $A$ and $B$, satisfies:
			\begin{equation*}
			d_p(A,B)\leq c\,||f - g ||_{\infty}^{1-\frac{k}{p}}\;\;
			\mbox{ for every $p\ge k$.}
			\end{equation*}		
		\end{theorem}
		When $p=\infty$, the constant $c$ is no longer necessary, obtaining the following most commonly used simplified version.
		\begin{corollary}[\mbox{\cite[p.~183]{Computational}}]\label{cor:d}
			Let $K$ be a simplicial complex and  $f,g : K \rightarrow \RR$ be two monotonic functions. If $A$ and $B$ denote 
			the corresponding
				persistence barcodes obtained from $f$ and $g$,   then
			\begin{equation*}
			d_\infty(A,B) \leq ||f-g||_\infty.
			\end{equation*}
		\end{corollary}
		
		Finally, as a consequence of Theorem \ref{Est1}, 
		we can assert the following.
		
		\begin{theorem}[\cite{Chazal-figure}]\label{Est2}
			Consider two finite metric spaces $(X,d_X)$, $(Y,d_Y)$. Let $A,B$ be the two 
			persistence barcodes obtained, respectively, from $ Rips(X,t)|_{t\in\RR} $ and  $Rips(Y,t)|_{t\in\RR}$.  Then,
			$$d_\infty(A,B) \leq d_{GH}(X,Y)$$
			where $d_{GH}$ denotes the Gromov-Hausdorff (GH) distance\footnote{The Gromov-Hausdorff distance between
$X$ and $Y$ is
 $\inf_{\gamma_X,\gamma_Y}$ $
d_H^Z(\gamma_X(X),\gamma_Y(Y))$
where
$d_H^Z(\gamma_X(X),\gamma_Y(Y))$ is the Hausdorff distance between $\gamma_X(X)$ and $\gamma_Y(Y)$
and 
$\gamma_X,\gamma_Y$  range over all the isometric embeddings of $X,Y$
into some same metric space $(Z,d_Z)$.
}.
		\end{theorem}

Looking at these results, we can conclude that stability results are simpler when using  the  bottleneck distance than when using  the  Wasserstein distance.
		

	\section{Stability of persistent entropy}
	
This section aims to show  under which conditions persistent entropy is stable, which means that it is uniformly continuous or,  more informally, there is a bound that ``controls'' the perturbation produced by noise in the input data.
	In the first subsection, we recall the definition of persistent entropy. Later, we provide several lemmas that will be needed to prove the stability of persistent entropy for finite persistence barcodes. Lastly, we will see how we can project persistence barcodes with infinite length intervals to finite persistence barcodes in a stable  way. These projections will allow to provide general stability results for persistent entropy. 
	
	\subsection{Persistent entropy}
	
So far, we have seen how persistent homology can be represented using persistence barcodes in a stable way. Nevertheless, sometimes, we might prefer to use only a number to summarize persistent homology (such as
 persistent entropy),  even if we are losing information by
 doing so.
	
	\begin{definition}[persistent entropy \cite{ Barcodes-entropy,Persistent-entropy}]\label{Entropy-def}
	The {\it persistent entropy} $E(A)$ of a persistence barcode $A = \big\{ [x^a_i,y^a_i)\big\}_{1\leq i\leq n_a}$ in 	$\mathcal{B}_F$ is defined as:
		\begin{equation*}
		E(A) = - \suma 
		\frac{\ell_i^a}{L_a}\log\left( \frac{\ell_i^a}{L_a} \right).
		\end{equation*}  
	\end{definition}  

For simplicity of notation, $\log$ will refer to the $\log$-base-2 function.
Observe that, to compute persistent entropy, we only have to consider the  length $\ell_i^a$ of each interval
$[x^a_i,y^a_i)$. 		
The following immediate result holds.
		\begin{remark}\label{ent_pro}
			If $A\in \mathcal{B}_F$ then
			$E(\psi(A)) = E(A)$.
		\end{remark}
	
	Let us see now a naive example of application of persistent entropy.

	\begin{example}
		Suppose we have 20 point clouds: 10 point clouds following a normal distribution and 10 point clouds following a uniform distribution (see Figure~\ref{figure:example}). 
		Note that since the sample is small, we should not perform a multivariate statistical test, so the idea is to perform  univariate statistical tests using persistent entropy.
	Let us compute the $1$-st persistent homology  using the Vietoris-Rips filtration. Observe that the computed persistence barcodes will never have  infinite length intervals since Vietoris-Rips complexes are always contractible from a (large enough) value. 
		Now,  let us compute the persistent entropy of each persistence barcode to obtain a number for each of the point clouds. Let us set $\alpha = 0.05$ and perform the Mann–Whitney U test\footnote{See \cite{statistics_web} for a simple introduction to statistical tests.}. We obtain a  $p$-value  of $p=0.046$ for this experiment so $p<\alpha$ and we can conclude that there are significant differences between the point clouds.
		
        \begin{figure}[ht]
        \centering
        \includegraphics[width = .75\textwidth]{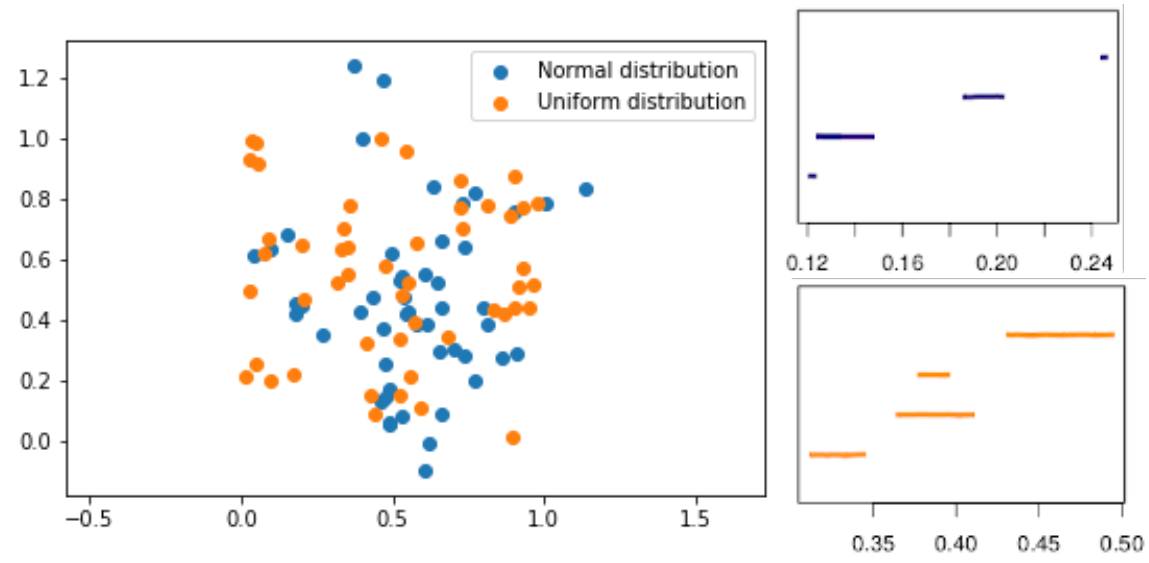}
			\caption{Left: in blue, a point cloud $X$ following a normal distribution; in orange, a point cloud $Y$ following a uniform distribution. Right: 1-th persistence barcode  of the Vietoris-Rips filtration associated to $X$ (top) and to $Y$  (bottom).}\label{figure:example}
        \end{figure}
	\end{example}
	Note that, in the definition of persistent entropy, we assume that there are no infinite length intervals in the persistence barcode.  We will study  in Subsection \ref{subsec:infinity} how to proceed when 
  infinite length intervals appear.

	\subsection{Preliminary lemmas}
	
	In this subsection, we will provide several results  useful to prove the main results in this paper that will be given in Subsection~\ref{subsec:stabilityresults}.

	Let us  recall a well-known result regarding  $p$-norms.
	\begin{remark}\label{p-inequality}
		Let $z\in \mathbb{R}^n$ and $p,q\in\mathbb{R}$. 
		Let $||z||_p=\left( \sum_{i=1}^n |z_i|^p \right)^{\frac{1}{p}}$
		and $||z||_\infty = \max_i \{ |z_i| \}$. 
		If $1 \leq q < p \leq \infty$ then 
		$||z||_p \leq ||z||_q  \leq  n^{\frac{1}{q} - \frac{1}{p}} ||z||_p$.
	\end{remark}
	
The following result extends  Remark \ref{p-inequality} 	to  the Wasserstein distance.

	\begin{lemma}\label{w-inequality}
		Let $d_p$ be the $p$-th Wasserstein distance for persistence barcodes. If $A,B \in \mathcal{B}_F$ and $1 \leq q < p \leq \infty$ then
		$$d_p(A,B) \leq d_q(A,B) \leq (n_p)^{\frac{1}{q}-\frac{1}{p}}d_p(A,B).$$
	\end{lemma}

	\begin{proof}
		For $r=p,q$, let $\gamma_r$ denote a bijection 
		where $d_r(A,B)$ is reached, that is,
		\begin{equation*}
		d_r(A,B)=\Big(\sumri \max \big\{ |x_i^a - x_{\gamma_r(i)}^b|^r, |y_i^a - y_{\gamma_r(i)}^b|^r\big\}\Big)^{1/r}=\Big(\sumri \big(\zabri\big)^r\Big)^{1/r}
		\end{equation*}
	where	$n_r$ is the cardinality of  $\gamma_r$ and $\zabri=\max \big\{ |x_i^a - x_{\gamma_r(i)}^b|, |y_i^a - y_{\gamma_r(i)}^b|\big\}$. 
Observe that, by definition of  $d_q$ and applying Remark~\ref{p-inequality}, we have
\begin{equation}
d_q(A,B) \leq \left(\sumi \big(\zabpi\big)^q\right)^{\frac{1}{q}}\leq (n_p)^{\frac{1}{q} - \frac{1}{p}} \left(\sumi \big(\zabpi\big)^p \right)^{\frac{1}{p}}. \end{equation} 
	Therefore, $d_q(A,B)\leq (n_p)^{\frac{1}{q} - \frac{1}{p}}d_p(A,B)$.
	Besides, by definition of  $d_p$ and again by Remark~\ref{p-inequality},
\begin{equation*}
d_p(A,B) 
\leq
		\left(\sum_{i=1}^{n_q} \big(
	\zabqi\big)^p
		\right)^{\frac{1}{p}}
		\leq \left(\sum_{i=1}^{n_q} \big(\zabqi\big)^q \right)^{\frac{1}{q}}=d_q(A,B)
		\end{equation*}
		  concluding that
		 $d_p(A,B)\leq d_q(A,B)$.			
	\end{proof}
	
The  result below states that when we translate the intervals of  given persistence barcodes $A$ and $B$ to the origin by projection $\pi$, the distance between them can be doubled.

	\begin{lemma}\label{0-inequality}
		If $A,B \in \mathcal{B}_F$ 
		then
	$$d_p(\pi(A),\pi(B)) \leq 2\,d_p(A,B).$$
	\end{lemma}
	\begin{proof}
Let $\gamma_p$ be a bijection where $d_p(A,B)$ is reached.
Let $n_p$ denote the cardinality of $\gamma_p$.
Since  $\pi(A) = \{ [0, 
		\ell_i^a
		) \}_{1\leq i\leq n_a}$ and $\pi(B) = \{[0, 
		\ell_{i}^b
		) \}_{1\leq i\leq n_b}$ then we have:
		\begin{align*}
		&\big(d_p (\pi(A),\pi(B))\big)^p= \min_{\gamma} \sum_{i=1}^{n_{\gamma}}
		\max\big\{ 0,|
			\ell_i^a-	\ell_{\gamma(i)}^b
		|^p  \big\}
		=  \min_{\gamma}\sum_{i=1}^{n_{\gamma}} |
		\ell_i^a-	\ell_{\gamma(i)}^b
		|^p  \\
		&\leq \sum_{i=1}^{n_p} |
		\ell_i^a-	\ell_{\gamma_p(i)}^b
		|^p 
		\leq \sumi \Big( |x_i^a-x_{\gamma_p(i)}^b|+|y^a_i-y^b_{\gamma_p(i)}| \Big)^p 
		\\&
				\leq \sumi  \big( 2\max\big\{ |x_i^a - x_{\gamma_p(i)}^b|, |y_i^a - y_{\gamma_p(i)}^b|\big\} \big)^p
				= 2^p \big(d_p(A,B)\big)^p.
		\end{align*}
	\end{proof}

	  To establish what we consider ``big'' or ``small'' error, we need to normalize the distances between persistence barcodes  in some way.
	
	\begin{definition}[relative error]
	 Let $A,B \in \mathcal{B}_F$ and  $1\leq p \leq \infty.$
	The {\it relative 	error} 	$r_p(A,B)$ 	is defined as:
		\begin{equation*}\label{error}
		r_p(A,B) = 
		\frac{2(n_p)^{1-\frac{1}{p}}}{L_{max}}\; d_p(A,B).
		\end{equation*}
	\end{definition}
Observe  that, according to Lemma~\ref{0-inequality},  it is  satisfied that
$$d_p(\pi(A), \pi(B)) \leq \frac{L_{max}}{(n_p)^{1-\frac{1}{p}}}\; r_p(A,B).$$
The next lemma is a technical result that we will use  later.
\begin{lemma}\label{lema-tecnico}
	Let $\gamma$ be a bijection between the multisets $A,B \in \mathcal{B}_F$. Let $n_{\gamma}$ be the cardinality of $\gamma$. Then 	for all $i$,  $1\leq i\leq n_\gamma$, we have:
		  $$
		\left| \dfrac{\ell_i^a}{L_a} - \dfrac{\ell^b_{\gamma(i)}}{L_b} \right|
			\leq
		\frac{ \big|\ell^a_i - \ell_{\gamma(i)}^b\big|}{
		L_{\max}
		} + \frac{ \ell^b_{\gamma(i)}d_1(\pi(A),\pi(B))}{L_a L_b}.
		$$
	\end{lemma}
	
	\begin{proof}
	Without loss of generality, 
	suppose $L_{\max}=L_a$. Since
	\begin{align}\label{ec1}\left| \dfrac{\ell_i^a}{L_a} - \dfrac{\ell^b_{\gamma(i)}}{L_b} \right|
			=
		\left| \dfrac{\ell_i^aL_b - \ell^b_{\gamma(i)}L_a}{L_aL_b} \right|
		\end{align}
		we consider two cases:
		$\ell_i^aL_b \geq \ell^b_{\gamma(i)}L_a$ and 	$\ell_i^aL_b \leq \ell^b_{\gamma(i)} L_a$. In the first case:
	\begin{equation*}\label{l1eq}
	(\ref{ec1})	=  \frac{\ell_i^aL_b - \ell^b_{\gamma(i)}L_a}{L_a L_b}  \leq  \frac{\ell_i^aL_b - \ell^b_{\gamma(i)}L_b}{L_a L_b} =  \frac{\ell_i^a - \ell^b_{\gamma(i)}}{L_a}.
		\end{equation*}
		For the second case  (i.e., when $\ell_i^aL_b \leq \ell^b_{\gamma(i)}L_a$), use that 
		$L_a\leq L_b + d_1(\pi(A), \pi(B))$ 
		to obtain:
		\begin{align*}
		(\ref{ec1})
		&= \frac{ \ell^b_{\gamma(i)}L_a - \ell_i^aL_b}{L_a L_b} 
		\leq  \frac{ \ell^b_{\gamma(i)} \left( L_b + d_1(\pi(A),\pi(B)) \right) - \ell_i^aL_b}{L_a L_b} \\
		&=  \frac{ \ell^b_{\gamma(i)} - \ell_i^a}{L_a} +  \frac{ \ell^b_{\gamma(i)}d_1(\pi(A),\pi(B))}{L_a L_b}.
		\end{align*}
	\end{proof}
	
	Let us now see  how the projection $\psi$ affects the relationship between the 
relative error $r_p$ and the
distance $d_1$.
	\begin{lemma}\label{Projection-inequality}
		If $A,B \in \mathcal{B}_F$ and $1\leq p \leq \infty$ then 
		  $$
		d_1(\psi(A),\psi(B)) \leq
		2\,r_p(A,B).
		$$
	\end{lemma}
	\begin{proof}
Recall that if 
		$A=\big\{\left[x_i^a,y_i^a\right)\big\}_{1\leq i\leq n_a}$ and $B=\big\{\left[x_i^b,y_i^b\right)\big\}_{1\leq i\leq n_b}$ then we have:
		$$\pi(A)=\big\{\left[0,\ell_i^a\right)\big\}_{1\leq i\leq n_a}\quad\mbox{ and }\quad \pi(B)=\big\{\left[0,\ell_i^b\right)\big\}_{1\leq i\leq n_b},$$
		$$\psi(A)=\left\{\left[0,\dfrac{\ell_i^a}{L_a}\right)\right\}_{1\leq i\leq n_a}\quad\mbox{ and }\quad \psi(B)=\left\{\left[0,\dfrac{\ell_i^b}{L_b}\right)\right\}_{1\leq i\leq n_b}.$$
		Let $\gamma_{\pi,1}$ be a bijection 
		where $d_1\big(\pi(A),\pi(B)\big)$ is reached, that is:
		\begin{equation*}
		d_1\big(\pi(A),\pi(B)\big)=\sum_{i=1}^{n_{\pi,1}}
		\big|\ell_i^a-\ell_{\gamma_{\pi,1}(i)}^b\big|
		\end{equation*}
		where $n_{\pi,1}$ is the cardinality of  $\gamma_{\pi,1}$. Notice that  $\ell_i^a$ or $\ell_{\gamma_{\pi,1}(i)}^b$ might be $0$  for some $i$ if intervals of zero length were needed for creating  bijection $\gamma_{\pi,1}$.  
		We can assume without loss of generality that $L_{\max} = L_a$. Now by Lemma~\ref{lema-tecnico} we have:
			\begin{align*}
		d_1\big(\psi(A),\psi(B)\big) 
		&\leq \sum_{i=1}^{n_{\pi,1}}
			\left| \dfrac{\ell_i^a}{L_a} - \dfrac{\ell^b_{\gamma_{\pi,1}(i)}}{L_b} \right|\\
&\leq \sum_{i=1}^{n_{\pi,1}} \left( \frac{ \big|\ell^a_i - \ell_{\gamma_{\pi,1}(i)}^b\big|}{L_a} + \frac{ \ell^b_{\gamma_{\pi,1}(i)}d_1(\pi(A),\pi(B))}{L_a L_b} \right)  \\
		&= \frac{d_1(\pi(A),\pi(B))}{L_a}  + \frac{ L_b d_1(\pi(A),\pi(B))}{L_a L_b} = \frac{2d_1(\pi(A),\pi(B))}{L_a}.
		\end{align*}
		Applying Lemma~\ref{0-inequality} we have that
		$
		\dfrac{2d_1(\pi(A),\pi(B))}{L_a} \leq \dfrac{4d_1(A,B)}{L_a}$.\\
			By Lemma~\ref{w-inequality}, we get that
			$
	\dfrac{4d_1(A,B)}{L_a}	\leq \dfrac{ 4(n_p)^{1-\frac{1}{p}}  d_p(A, B)}{L_a}$.\\
			Finally, since we assumed that $L_a=L_{\max}$ then 
		$\dfrac{ 4(n_p)^{1-\frac{1}{p}}  d_p(A, B)}{L_a} = 2r_p(A,B)$.
	\end{proof}
	
	
\subsection{Stability results for $\mathcal{B}_F$}\label{subsec:stabilityresults}
	
Two important results about the stability of persistent homology were recalled in Section 2 (Theorem \ref{Est1} and Theorem \ref{Est2}). 
These results  guarantee that if two filter functions (or two metric spaces) are ``similar'', then  their corresponding persistence barcodes will be ``similar'' as well. 
Besides, there also exist stability results for Shannon entropy defined on  probability distributions. To combine these results to prove stability of persistent entropy we need to adapt the last ones to the metric space of persistence barcodes.
	
First of all, recall that the continuity of persistent entropy with respect to the bottleneck distance is proven in \cite{Noise}.
The following proposition generalizes that result to the  Wasserstein distance.
	
	\begin{proposition}\label{continuity}
		Let $A,B \in \mathcal{B}_F$ and let $d_p$ be the $p$-th Wasserstein distance with 
	$1\leq p \leq \infty$. 
		If we fix  a maximum number of intervals and a minimum  sum of the lengths of the intervals in a persistence barcode, then the persistent entropy $E$ is continuous on $(\mathcal{B}_F,d_p)$:
		\begin{equation*}
		\forall \ee \; \exists \delta \text{ such that }  d_p(A,B) \leq \delta \then |E(A) - E(B)|\leq \ee.
		\end{equation*}
	\end{proposition}
	
	\begin{proof}
We have that 	$d_\infty(A,B)\leq d_p(A,B)$ by
	Lemma~\ref{w-inequality}. Since $d_p(A,B)\leq \delta$ then 
	$d_\infty(A,B)\leq \delta$ and  by \cite[Proposition 1]{Noise} 
	\begin{equation*}
	d_{\infty}(A,B) \leq \delta \then |E(A) - E(B)|\leq \ee,
	\end{equation*}
	concluding the proof.
	\end{proof}

The stability of Shannon entropy has been previously studied by Lesche in \cite{Lesche} for the $1$-norm due to its importance in physics. That bound can be slightly improved as shown in \cite{ Elements}.
	
	\begin{theorem}[\mbox{\cite[p.~664]{Elements}}] \label{prob_ineq}
		Let $P$ and $Q$ be  two finite probability distributions (seen as vectors in $\RR^{u}
		$), and let $E_S(P)$ and $E_S(Q)$ be, respectively, their  Shannon entropy. If $||P-Q||_1 \leq \frac{1}{2}$ then 
		$$
		|E_S(P)-E_S(Q)| \leq 
		||P-Q||_1\big(\log({u}
		) - \log(||P-Q||_1)\big).
		$$
	\end{theorem}
	Notice that the restriction $||P-Q||_1 \leq \frac{1}{2}$ is reasonable because $||P-Q||_1$ is at most $2$. 
	
	Now, let us introduce one of the main result of this paper. We can observe that since the space $\mathcal{B}_0 \cap \mathcal{B}_N$ can be interpreted as finite probability distributions, we can first project the  persistence barcodes of $\mathcal{B}_F$ onto $\mathcal{B}_0 \cap \mathcal{B}_N$ and then apply the previous theorem to obtain the desired stability result.
	
	\begin{theorem}[stability of persistent entropy]\label{main_theo0}
		Let $A,B \in \mathcal{B}_F$. 	Let us assume that	$r_p(A,B)\leq \frac{1}{4}$. Then: 
	$$
		|E(A)-E(B)| 
        \leq 2r_p(A,B)\big(\log
        (n_a+n_b) - \log
        (2r_p(A,B))\big).
        $$ 
     \end{theorem}

	\begin{proof}
		First, 
	by Remark~\ref{ent_pro}, we have that 
	\begin{align}\label{ec2}|E(A)-E(B)|=|E(\psi(A))-E(\psi(B))|.\end{align}
Now, 
let $\gamma_{\psi,1}$ be a bijection where $d_1(\psi(A), \psi(B))$ is reached, that is,
	 $$d_1(\psi(A), \psi(B))=\sum_{i=1}^{n_{\psi,1}}
	 \left|\frac{\ell_i^a}{L_a}-\frac{\ell_{\gamma_{\psi,1}(i)}^b}{L_b}\right|$$
	 where $n_{\psi,1}$ is the cardinality of $\gamma_{\psi,1}$.
		    Let $P$ be the vector $\left(\frac{\ell_1^a}{L_a},
	    \dots,\frac{\ell_{n_{\psi,1}}^a}{L_a}\right)$
	    and $Q$ the vector 
	 $\left(\frac{\ell_{\gamma_{\psi,1}(1)}^b}{L_b},
	    \dots,\frac{\ell_{\gamma_{\psi,1}(n_{\psi,1})}^b}{L_b}\right)$.  Then,
	  $||P-Q||_1=d_1(\psi(A), \psi(B))$
	   and by  Lemma~\ref{Projection-inequality}, 
	    $$||P-Q||_1=d_1(\psi(A), \psi(B))\leq  
	    2r_p(A,B).
	    $$
	    Now,     
	    since $r_p(A,B)\leq \frac{1}{4}$ then 
	    $||P-Q||_1\leq \frac{1}{2}$. By Theorem~\ref{prob_ineq}
	    we have that
	    	  \begin{equation}\label{eq:E}
	        (\ref{ec2})
	        \leq d_1(\psi(A), \psi(B))\big(\log
	        (n_{\psi,1}) - \log
	        d_1(\psi(A), \psi(B))\big)
	        \end{equation}
	   Now, since
	     $x\big(\log(n_{\psi,1}) - \log(x)\big)$
	    is increasing as long as $x\leq \frac{n_{\psi,1}}{e}$ and
	   $d_1\big(\psi(A)$, $\psi(B)\big)\leq  
	    2r_p(A,B)\leq \frac{1}{2}\leq \frac{n_{\psi,1}}{e}$ since $n_{\psi,1}\geq 2$ by assumption\footnote{See the comment after Notation~\ref{notation}.}, then  
	   \begin{align}\label{ec3}\mbox{(\ref{eq:E})}
	   \leq 2r_p(A,B)\big(\log(n_{\psi,1}) - \log(2r_p(A,B))\big)
	   \end{align}
        Finally, 
        $$ (\ref{ec3})
        \leq  2r_p(A,B)\big(\log
        (n_a+n_b) - \log
        (2r_p(A,B))\big)$$
        since $n_{\psi,1}\leq n_a+n_b$.
\end{proof}
	
	 Although the bound of $|E(A)-E(B)|$ can tend to $\infty$ for an arbitrary large 
	 $n$ for $n=n_a+n_b$,
	 the relative value 
	$\frac{|E(A)-E(B)|}{\log(n)}$ is bounded when $n$ tends to $\infty$ since $r_p(A,B)\leq \frac{1}{4}$. In other words,
	\begin{equation*}
	\lim_{n \rightarrow \infty}\; \sup_{\mathcal{B}_F} \left( \frac{|E(A) - E(B)|}{\log(n)} \right) = 2 r_p (A,B).\nn
	\end{equation*} 
	Table \ref{table:table1} shows some numerical examples regarding such relative value.
	
	\begin{table}[th]
	\caption{Bounds $\frac{2r_\infty(A,B)\big(\log
        (n) - \log
        (2r_\infty(A,B))\big)}{\log(n)}$ of relative values $\frac{|E(A)-E(B)|}{\log (n)}$ for different values of $n$  (columns) and relative errors $r_\infty(A,B)$ (rows).}\label{table:table1}
        \medskip
		\begin{center}
		    \scriptsize
			\begin{tabular}{|c||c|c|c|c|}
				\hline
				\rowcolor[gray]{0.8}
				& \multicolumn{4}{|c|}{\textbf{
				$r_{\infty}(A,B)$}}\\
				\hline 
				\cellcolor[gray]{0.8}$\mathbf{n}$&$\mathbf{0.1}$ & $\mathbf{0.05}$ & $\mathbf{0.025}$ & $\mathbf{0.01}$ \\
				\hline
				\hline
				\textbf{10} & 0.339794 & 0.2 & 0.115051 & 0.0539794 \\
				\textbf{510} & 0.251631 & 0.136933 & 0.0740258 & 0.0325498 \\
				\textbf{1010} & 0.246531 & 0.133285 & 0.0716526 & 0.0313102 \\
				\textbf{1510} & 0.243975 & 0.131457 & 0.070463 & 0.0306888 \\
				\textbf{2010} & 0.242321 & 0.130274 & 0.0696935 & 0.0302868 \\
				\textbf{2510} & 0.24112 & 0.129415 & 0.0691346 & 0.0299949 \\
				\textbf{3010} & 0.240187 & 0.128747 & 0.0687007 & 0.0297682 \\
				\textbf{3510} & 0.239431 & 0.128206 & 0.0683486 & 0.0295843 \\
				\textbf{4010} & 0.238798 & 0.127754 & 0.0680541 & 0.0294305 \\
				\textbf{4510} & 0.238256 & 0.127366 & 0.067802 & 0.0292988 \\
				\textbf{5010} & 0.237784 & 0.127028 & 0.0675823 & 0.029184 \\
				\hline
			\end{tabular}
		\end{center}
	\end{table}
	
	
	\subsection{Persistence barcodes with infinite length intervals}\label{subsec:infinity}
	
	In order to extend the definition of persistent entropy to persistence barcodes with infinite length intervals, it is common to define a projection from $\mathcal{B}$ to $\mathcal{B}_F$ that transforms  infinite length intervals into finite length intervals. There are many ways to do this
	and depending on choice, persistent entropy may no longer be stable or scale-invariant.
	In this section, we explain some projections and their properties.
	
	We start with a simple example. 
	To avoid  calculations involving the infinite value
	when computing persistent homology, usually, an upper bound is fixed and considered to be the infinite value. Then, if we want to compute persistent entropy, the first idea could be just to assign this upper bound to each of the infinite values that appear in the  infinite length intervals. 
	 
 	\begin{definition}[projection $\xi_c$] 
 		Let $c\in \RR$. Define the projection 
 		$\xi_c: {\cal B}\to {\cal B}_F$ such that  
 		for  $A = \{[x^a_i, y^a_i)\}\in {\cal B}$,
 		 $$\xi_c(A) = \{ [x^a_i, z^a_i) \}\;\mbox{ where $z_i^a = 
 		c$ if $y_i^a = \infty$ and $z_i^a=y_i^a$ otherwise.}$$
 	\end{definition}
  	The following result confirms that the projection $\xi_c$ is stable.
 	
	\begin{proposition}\label{project_inf}
		Let  
		$A,B \in \mathcal{B}$. Then, projection $\xi_c$ satisfies
		that
		$$
		d_p(\xi_c(A),\xi_c(B)) \leq d_p(A,B).$$
	\end{proposition}
	
	\begin{proof}
	Let $A=\{[x_i^a,y_i^a)\}$ and $B=\{[x_i^b,y_i^b)\}$. 
Let $\gamma_p$ be a bijection 
where $d_p(A,B) $ is reached. 
Let  $n_p$ denote the cardinality of $\gamma_p$.
		Observe that if 
		$y_i^a<\infty$
		and $y_{\gamma_p(i)}^b<\infty$
		then $|z_i^a-z_{\gamma_p(i)}^b|=|y_i^a-y_{\gamma_p(i)}^b|$.
			Nevertheless, if $y_i^a=\infty$ and   $y_{\gamma_p(i)}^b<\infty$ (resp.
		$y_i^a<\infty$ and   $y_{\gamma_p(i)}^b=\infty$)
		then $|z_i^a-z_{\gamma_p(i)}^b|<|y_i^a-y_{\gamma_p(i)}^b|=\infty$.
		Finally, if
		$y_i^a=y_{\gamma_p(i)}^b=\infty$ then 
		$|z_i^a-z_{\gamma_p(i)}^b|=|y_i^a-y_{\gamma_p(i)}^b|=0$.
		We conclude:
			\begin{align*}
			\big(d_p(\xi_c(A),\xi_c(B))\big)^p &=\min_{\gamma}\sum_{i=1}^{n_{\gamma}} \max\{|x_i^a-x_{\gamma(i)}^b|^p,|z_i^a-z_{\gamma(i)}^b|^p\}\\
			&\leq
		\sumi \max\{|x_i^a-x_{\gamma_p(i)}^b|^p,|z_i^a-z_{\gamma_p(i)}^b|^p\}
	\leq \big(d_p(A,B)\big)^p.
		\end{align*}
	\end{proof}	

	Despite being stable, $\xi_c$ is not scale-invariant. 
	  By definition, a projection $f: {\cal B}\to {\cal B}_F$ is scale-invariant if $f(\lambda A) = \lambda f(A)$, being $\lambda A$  the scalar multiplication of each of the intervals (notice that $\lambda\cdot\infty = \infty$).
	  We now define the following stable and scale-invariant projections from ${\cal B}$ to ${\cal B}_F$.

	 \begin{definition}[projections $\mu_{\lambda}$, $\nu_{\lambda,p}$, $\tau_{\lambda}$]\label{projection_definition}
	 	Let $\lambda \geq 0$ and $1\leq p \leq \infty$.
	 	Let $A= \{[x_i^a,y_i^a) \}\in\mathcal{B}$.  Then:
	 	\begin{itemize}
						\item $\mu_\lambda(A) = \{ [x^a_i, z^a_i) \}$
			where
		$z_i^a = 
			 x_i^a + \lambda\ell_{max}^a$ if $y_i^a = \infty$ and $z_i^a=y_i^a$ otherwise; being $\ell_{max}^a$  the maximum finite value for $\ell_i^a=y_i^a-x^a_i$.
			\item $\nu_{\lambda,p}(A) = \{ [x^a_i, z^a_i) \}$ where $z_i^a = 
			x_i^a + \lambda L_{a,p}$ if $y_i^a = \infty$ and $z_i^a=y_i^a$ otherwise; being 
	$	L_{a,p}=
	\left(\sum_{i\in I} (\ell_i^a )^p \right)^{1/p}$ where $I=\{i:\,1\leq i\leq n_a$ and $\ell_i^a <\infty\}$.
			\item 
		$\tau_\lambda(A) = \{ [x^a_i, z^a_i) \}$ where $z_i^a = 
			 (1+\lambda) u_a$ if $y_i^a = \infty$ 
			 	 and $z_i^a=y_i^a$ otherwise;
			 	being $u_a$ the maximum finite value for 			$y_i^a$.
	 	\end{itemize}
	 \end{definition}
	 Notice that $\mu_0 = \nu_{0,p}$
	 and both are equivalent to remove the infinite length intervals.

	 \begin{proposition}[stability of projections $\tau_{\lambda}$, $\mu_{\lambda}$, $\nu_{\lambda,p}$]\label{projection-stability}
	 	Given two persistence barcodes $A,B\in {\cal B}$ with the same number $m$ of  infinite length intervals, we have that:	 		 	\begin{align*}&d_p(\mu_\lambda(A),\mu_\lambda(B)) 
	 		\leq \big( 1+m2^p\lambda^p \big)^{1/p}  d_p(A,B);\\
	 		 		&d_p(\nu_{\lambda, p}(A),\nu_{\lambda, p}(B)) \leq \left( 1+m2^p\lambda^p \right)^{1/p} d_p(A,B).
	 		 		\end{align*}
	 		 		
If the length of the longest finite interval in $A$ and $B$ are both greater than $2 d_\infty(A,B)$, then
$$d_p(\tau_\lambda(A),\tau_\lambda(B)) 
	 		 		\leq \big(  1+m\left( 1+\lambda  \right)^p\big)^{1/p} d_p(A,B).
$$
	 \end{proposition}
	 
	 \begin{proof}
	 Sort the intervals of $A$ and $B$ such that 
	 	their first $m$ intervals are the infinite length intervals 
	 	and for $r=p,\infty$, consider a bijection $\gamma_r$ 
	 	where $d_r(A,B) $ is reached.
	 	Let $n_r$ denote the cardinality of $\gamma_r$.
	 	Let $f$ refer to $\tau_{\lambda}$, $\mu_{\lambda}$ or $\nu_{\lambda,p}$.
		 	We have:
	\begin{align*}
	 &\big(	d_p(f(A),f(B))\big)^p =  \min_{\gamma}\sum_{i=1}^{n_{\gamma}} \max \big\{ |x_i^a - x_{\gamma(i)}^b|^p, |z_i^a - z_{\gamma(i)}^b|^p \big\} \nn \\
	 	&\leq \sumi \max \big\{ |x_i^a - x_{\gamma_p(i)}^b|^p, |z_i^a - z_{\gamma_p(i)}^b|^p \big\} \\
	 &	= \sum_{i=1}^{m} \max \big\{ |x_i^a - x_{\gamma_p(i)}^b|^p, |z^a_i  - z^b_{\gamma_p(i)}|^p \big\} + \sum_{i=m+1}^{n_p} \max \big\{ |x_i^a - x_{\gamma_p(i)}^b|^p, |y_i^a - y_{\gamma_p(i)}^b|^p \big\} \nn \\
	 	&=  \sum_{i=1}^{m} \max \big\{ |x_i^a - x_{\gamma_p(i)}^b|^p, |z^a_i - z^b_{\gamma_p(i)}|^p \big\} + \big(d_p(A,B)\big)^p  - \sum_{i=1}^{m}  |x_i^a - x_{\gamma_p(i)}^b|^p \\
	 	&= \sum_{i=1}^{m} \max \big\{ 0, |z_i^a - z_{\gamma_p(i)}^b|^p - |x_i^a - x_{\gamma_p(i)}^b|^p \big\} + \big(d_p(A,B)\big)^p\\
	 		&=\sum_{i=1}^{m} \big(|z_i^a - z_{\gamma_p(i)}^b|^p - |x_i^a - x_{\gamma_p(i)}^b|^p  \big)+ \big(d_p(A,B)\big)^p.
	 	\end{align*}
	 	If $f = \mu_\lambda$ then, for all $i$, $1\leq i\leq m$, we have:
	 	\begin{align*}
	 		 	&|z^a_i - z^b_{\gamma_p(i)}|^p  - |x_i^a - x_{\gamma_p(i)}^b|^p = |\lambda \ell_{max}^a - x_i^a - \lambda\ell_{max}^b + x_{\gamma_p(i)}^b|^p  - |x_i^a - x_{\gamma_p(i)}^b|^p \\
	 		 	&\leq |x_i^a- x_{\gamma_p(i)}^b|^p + |\lambda\ell_{max}^a- \lambda\ell_{max}^b|^p  - |x_i^a - x_{\gamma_p(i)}^b|^p = \lambda^p|\ell_{max}^a- \ell_{max}^b|^p.
	 		 	\end{align*} 
Assume, without loss of generality, that $\ell_{max}^a \geq \ell_{max}^b$.
Then, the interval with length $\ell^b_*$ paired to an interval with length $\ell_{max}^a$  by bijection $\gamma_{\infty}$  satisfies, by definition, that $\ell^b_*\leq \ell_{max}^b \leq \ell_{max}^a$ and then	 
	 		 \[|\ell_{max}^a - \ell_{max}^b| \leq |\ell_{max}^a - \ell^b_*|\leq
	 		 	 2d_\infty(A,B), \] obtaining
	 		 	 	\[ |z^a_i - z^b_{\gamma_p(i)}|^p  - |x_i^a - x_{\gamma_p(i)}^b|^p  \leq 2^p\lambda^p \big(d_\infty(A,B)\big)^p \]
	 		 	and
	 		 	\begin{align*}
	 		 	    	 & \sum_{i=1}^{m} \big(|z_i^a - z_{\gamma_p(i)}^b|^p - |x_i^a - x_{\gamma_p(i)}^b|^p  \big)+ \big(d_p(A,B)\big)^p\\ &\leq  m2^p\lambda^p(d_\infty(A,B))^p + \big(d_p(A,B)\big)^p \leq ( m2^p\lambda^p+1)\big(d_p(A,B)\big)^p.
	 		 	\end{align*}
 		 		If $f = \nu_{\lambda, p}$ then,
 		 		for all $i$, $1\leq i\leq m$, we have: 
	\begin{align*}
 		 	&	|z^a_i - z^b_{\gamma_p(i)}|^p - |x_i^a - x_{\gamma_p(i)}^b|^p = \left| x_i^a + 
 		 		\lambda L_{a,p} 
 		 		 		- x_{\gamma_p(i)}^b - 
 		 		\lambda L_{b,p} 
 		 	 		 		\right|^p  - |x_i^a - x_{\gamma_p(i)}^b|^p \\
 		 		&\leq \left| 
 		 		\lambda L_{a,p} 
 		 	 		 		-\lambda L_{b,p} 
 		 	 		 		\right|^p + |x_i^a - x_{\gamma_p(i)}^b|^p - |x_i^a - x_{\gamma_p(i)}^b|^p = \lambda^p \left| 
 		 	L_{a,p} 
 		 				-L_{b,p} 
 		  		 		\right|^p.
 		 		\end{align*}
     		 	By the reverse triangle inequality:
 		 	$$
 		 	\lambda^p \left| 
 		 	L_{a,p} 
 		  		 	-L_{b,p}
 		  		 	\right|^p \leq 	
 		 	\lambda^p	\sum_{i=m+1}^{n_{\pi,p}}|\ell_ i^a-\ell_{\gamma_{\pi,p}(i)}^b|^p \leq
 		 	\lambda^p \big(d_p(\pi(A),\pi(B))\big)^p  		 		$$
 		 	being $n_{\pi,p}$ the cardinal of a bijection $\gamma_{\pi,p}$ where $d_p(\pi(A),\pi(B)) $ is reached.
By Lemma \ref{0-inequality},
 		 	$$\lambda^p\big(d_p(\pi(A),\pi(B))\big)^p \leq 2^p\lambda^p \big(d_p(A,B)\big)^p$$
 		 	and finally,
 $$	 \sum_{i=1}^{m} \big(|z_i^a - z_{\gamma_p(i)}^b|^p - |x_i^a - x_{\gamma_p(i)}^b|^p  \big)+ \big(d_p(A,B)\big)^p \leq ( m2^p\lambda^p+1)\big(d_p(A,B)\big)^p.
$$
If $f=\tau_\lambda$ then
	 	\begin{align*}
	 	&\sum_{i=1}^{m} \big( |z_i^a - z_{\gamma_p(i)}^b|^p - |x_i^a - x_{\gamma_p(i)}^b|^p \big) + \big(d_p(A,B)\big)^p\\
	 	&\leq \sum_{i=1}^{m} |z^a_i - z^b_{\gamma_p(i)}|^p + \big(d_p(A,B)\big)^p =
	 	(1+\lambda)^pm|u^a - u^b|^p
	 	+ \big(d_p(A,B)\big)^p.
	 	\end{align*}
	 	We only have to prove that $|u^a - u^b| \leq d_\infty(A,B)$. 
By reduction to the absurd, suppose that  $u^a - u^b > d_\infty(A,B)$.
	 		Without loss of generality, assume $u^a \geq u^b$.
	 	Take one interval $\alpha$ in $A$ with endpoint $u^a$ and another one 
	  in $B$ with endpoint $u^b$.
	 	Since, by hypothesis,  the length of both intervals is greater than $2d_\infty(A,B)$ then we can assume that they are not paired with the diagonal when computing the bottleneck distance. 
	  	Let $[x^b, y^b)$ be the interval in $B$ paired with $\alpha$. Then 
	 	\begin{equation*}
	 	u^a -  y^b \leq d_\infty(A,B) < u^a - u^b \then u^b < y^b
	 	\end{equation*}
	 leading to a contradiction. Therefore,
  	 	\begin{align*}
	 (1+\lambda)m |u^a - u^b|^p + \big(d_p(A,B)\big)^p 
 		 	& \leq  (1 + \lambda)^p m\big(d_\infty(A,B)\big)^p + \big(d_p(A,B)\big)^p \\
 		 	&\leq \left( (1 + \lambda)^pm+1\right) \big(d_p(A,B)\big)^p.
 	 	\end{align*}
	 \end{proof}
	 
Of course,  these projections are just a few of the many possible
that can be defined.
In the past,  since persistent entropy only takes into account the length of the intervals,  infinite length intervals were usually replaced by intervals of a fixed finite length. For example, in \cite{Piecewise}, $\tau_1$ was used plus a constant. With respect to this  case, notice that
	 adding a constant in the definition of any of the  projections above will produce  stable but not scale-invariant projections. In \cite{cells},  infinite length intervals were ignored using the stable and scale-invariant projection $\mu_0$ obtaining a topological based variable for analyzing cell arrangement.
	 
		\subsection{Stability results for ${\cal B}$} 
		
Let us now introduce the following results on the stability of
persistent entropy 
for the general case.
For simplicity, we have removed  infinite length intervals using $\mu_0$ for these statements, but we could use any other stable projection to remove such intervals.
	This way, the formulas that appear in the statements below would change according to the inequalities of  Proposition~\ref{projection-stability}.

\begin{theorem}\label{main_theo1}
			Let $K$ be a simplicial complex  and let $f,g : K \rightarrow \RR$ be two monotonic functions. Let $A,B \in \mathcal{B}$ be their corresponding persistence barcodes. 
		If $||f-g||_\infty \leq \frac{1}{8}\frac{L_{max}}{n} $
				then
					\begin{equation*}
			  |E(\mu_0(A)) - E(\mu_0(B))| \leq \frac{4 n ||f-g||_\infty }{ L_{\max}} \left( \log(n)  - \log \left( \frac{4 n ||f-g||_\infty }{ L_{\max}} \right) \right).
			\end{equation*}
		where $n=n_a + n_b$, being $n_a$ (resp. $n_b$)  the number of intervals of $\mu_0(A)$ (resp. $\mu_0(B)$).
			\end{theorem}
			\begin{proof}
		  First, 
		  	using 		    Corollary~\ref{cor:d}, we have that  $d_\infty(A,B) \leq ||f - g||_{\infty} $. 
Then,
$$r_{\infty}(\mu_0(A),\mu_0(B))=\frac{2d_{\infty}(A,B)n_\infty} {L_{max}}		   \leq \frac{2||f-g||_{\infty}n_\infty} {L_{max}}\leq \frac{1}{4}.
			$$
Therefore, by Theorem~\ref{main_theo0}, we have:
		   		 \begin{align*}
		|E(\mu_0(A)) - E(\mu_0(B))|
		 \leq 2r_{\infty}(\mu_0(A),\mu_0(B)) \left( \log(n
			)  - \log \left( 2r_{\infty}(\mu_0(A),\mu_0(B)) \right) \right).
			\end{align*}
			Since the function $x\big(\log(n)-\log(x)\big)$ is increasing as long as $x\leq\frac{n}{e}$
			and $\frac{1}{2}\leq \frac{n}{e}$ since $n\geq 2$ by assumption, 
			then
			\begin{align*}
			&2r_{\infty}(\mu_0(A),\mu_0(B)) \left( \log(n
			)  - \log \left( 2r_{\infty}(\mu_0(A),\mu_0(B)) \right) \right)\\
			&\leq
			\frac{4n ||f-g||_\infty  
			} {L_{max}} \left( \log(n
			)  - \log \left( \frac{4n ||f-g||_\infty 
			} {L_{max}} \right) \right).
		\end{align*}
		\end{proof}
		\begin{theorem}\label{main_theo2}
			Let $A,B$ be the persistence barcodes obtained respectively from $ Rips(X,t)|_{t\in\RR} $ and  $Rips(Y,t)|_{t\in\RR}$, being
			$(X,d_X)$ and $(Y,d_Y)$  two finite metric spaces.
		If $d_{GH}(X,Y) \leq \frac{1}{8} \frac {L_{max}}{n}$ then,	  
			\begin{equation*}
			  |E(\mu_0(A)) - E(\mu_0(B))| \leq \frac{4 nd_{GH}(X,Y)}{L_{max}} \left( \log(n)  - \log \left( \frac{4 n d_{GH}(X,Y)}{L_{max}} \right) \right).
			\end{equation*}
			where $n=n_a+n_b$, being $n_a$ (resp. $n_b$)  the number of intervals of $\mu_0(A)$ (resp. $\mu_0(B)$).
		\end{theorem}
		\begin{proof}
		Using  Theorem~\ref{Est2} we have that  $d_\infty(A,B) \leq d_{GH}(X,Y)$.
		As in the proof of Theorem~\ref{main_theo1} since 
		$d_{GH}(X,Y) \leq\frac{1}{8} \frac {L_{max}}{n}$ and the function $x\big(\log(n)-\log(x)\big)$ is increasing as long as $x\leq\frac{n}{e}$ then, by Theorem~\ref{main_theo0}, we obtain the desired result.
		\end{proof}
		
		It seems appropriate now to recapitulate the results of this section before moving on. As shown in the following diagram, at the beginning of the section we wanted to prove implication (A). In order to do it, we separated the problem into three parts ((1), (2) and (3)):
		
		\vspace{0.2cm}  
		\begin{center}
		\scriptsize
			\begin{tikzcd}
				\begin{tabular}{|c|}
					\hline
					Small 
					perturbations\\ in
					input data \\
					\hline
				\end{tabular}
				\arrow{rr}{(A)}&&
				\begin{tabular}{|c|}
					\hline
					Small perturbations\\ in
					persistent entropy \\
					\hline
				\end{tabular}
			\end{tikzcd}
		\end{center}	
		
	\begin{center}
	\scriptsize
			\begin{tikzcd}[column sep=tiny]
				\begin{tabular}{|c|}
					\hline
					Small \\
					perturbations\\ in 
					input data \\
					\hline
				\end{tabular}
				\arrow{rr}{(1)}&&
				\begin{tabular}{|c|}
					\hline
					Small perturbations\\
					in GH-distance\\
					or  filter function\\
					\hline
				\end{tabular}\arrow[d,"(2)"]\\
				&&
				\begin{tabular}{|c|}
					\hline
					Small \\
					perturbations \\
					in
					$(\mathcal{B}, d_p)$\\
					\hline
				\end{tabular}
				\arrow{rr}{(3)}&&
				\begin{tabular}{|c|}
					\hline
					Small 
					perturbations \\
					in persistent 
					entropy \\
					\hline
				\end{tabular}
			\end{tikzcd}
		\end{center}				
		\vspace{0.2cm}  
		Implication (1) is given by the formalization of the problem and implication (2) is given by Theorem \ref{Est1} and Theorem \ref{Est2} mentioned in the background section. The proof of implication (3) is the main aim of this section (Theorem \ref{main_theo0}). Putting all together we obtain Theorem \ref{main_theo1} and Theorem \ref{main_theo2}.

	  \section{Entropy-based summary functions}    
	
As we have already mentioned, numbers summarizing 
persistence barcodes  (such as persistent entropy)
are very useful to perform statistical tests. Nevertheless, if we want to perform a classification task, their discriminatory power might not be enough. One of the possible solutions is to 
summarize 
persistence barcodes using functions.
Summary functions (such as 
he already mentioned 
persistence silhouettes, 
Euler characteristic curves,
topological intensity maps 
or  
persistence landscapes)
have been used in the past to obtain statistical information from persistence barcodes. 
		For example, a simple but effective way of summarizing a persistence barcode is the {\it Betti  curve} defined as follows: If $A = \{[x_i^a, y_i^a)\}\in {\cal B}$ then	\[
			\beta(A)[t] = \mbox{cardinality of }
			\{ [x_i^a, y_i^a) :  x_i^a \leq t \leq y_i^a\}.
		\]
	That is, 
		$\beta(A)(t)$  is  the number of intervals in $A$ which are ``alive'' at time $t$.

In this section, we will define a  new summary piece-wise constant function (also known as step function).
It is  similar to the Betti curve but uses
		  persistent entropy instead of Betti numbers. We will prove its stability and show examples where such function measures different features of the persistence barcode than the Betti curve. Besides, and contrary to  what happened with persistent entropy, we will see that the normalization of this function is also stable. 
	
	\subsection{Entropy summary function (ES-function)}
	
	We now define a new function that pairs a persistence barcode $A \in \mathcal{B}_F$ with a real-valued piece-wise constant function.
	This new function summarizes information about the number of intervals of a given persistence barcode and their homogeneity and, as we will prove at the end of this subsection, is stable with respect to the bottleneck distance. 

	\begin{definition}[ES-function] 
		 The entropy summary function (ES-function) of a 
		  persistence  barcode $A = \{ [x_i^a, $ $y_i^a) \}_{1\leq i\leq n_a}$ in $\mathcal{B}_F$ is the 
		  real-valued
		  piecewise linear function:
		\begin{equation*}
		S(A)[t] =-\sum_{i=1}^{n_a} w^a_i(t) \frac{\ell_i^a}{L_a} \log \left( \frac{\ell_i^a}{L_a} \right)\nn
		\end{equation*}
		where
		$w_i^a(t) =1$ 	if $x_i^a\leq t \leq y_i^a$ and 
		$w_i^a(t) =0$ otherwise.         
	\end{definition}
	In other words, the ES-function  pairs  a persistence barcode $A=\{[x_i^a,y_i^a)\}$ 
		and an instant $t$
		with the partial sum of $E(A)$ corresponding to the intervals 
		$[x^a_i,y^a_i)$
		of $A$  that are ``alive'' at that moment $t$, that is, 
		 $x^a_i \leq  t \leq  y^a_i$.
		See Figure \ref{figure:ejemplo_ES}.
	Notice that $S(A): \mathbb{R} \rightarrow \mathbb{R}$ and  $S: \mathcal{B}_F \rightarrow \mathcal{C}$, being
	$\mathcal{C}$  the space of real-valued
	piece-wise constant functions.

\begin{figure}[!htb]
   	\begin{center}
   	    \includegraphics[width=.75\textwidth]{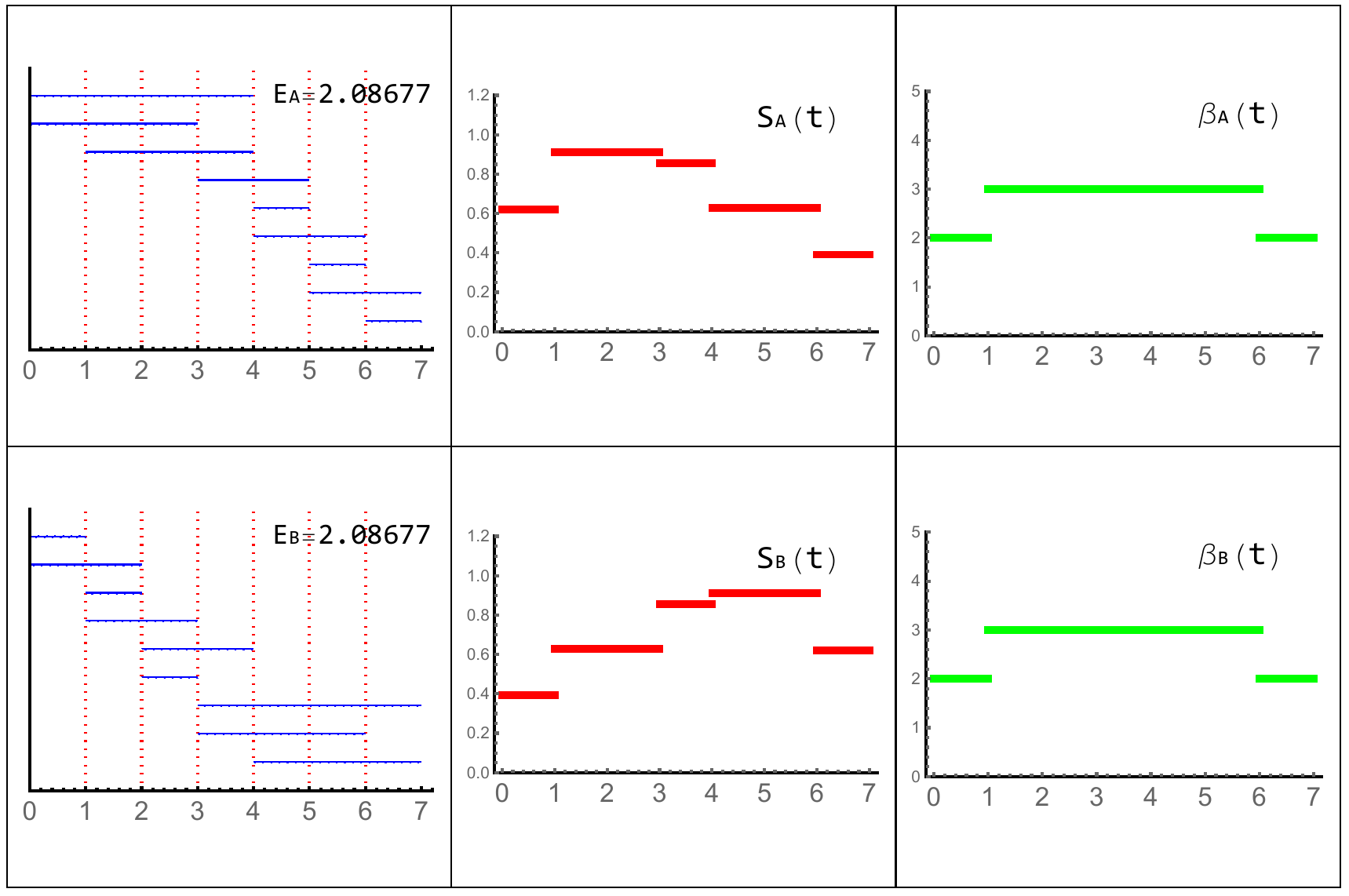}
   	\end{center}
 	\caption{In this example we can see two different persistence barcodes for which their  Betti  curves and persistent entropy are the same but not their ES-function.}\label{figure:ejemplo_ES}
      \end{figure}

	The following result states that the ES-function is stable with respect to the bottleneck distance.

	\begin{theorem}[stability of the ES-function]\label{Es-Stability}
		Let $S$ be the ES-function, $d_\infty$ the bottleneck distance and $A,B$  two persistence barcodes in $\mathcal{B}_F$. 
	Let $n_{\infty}$ be the cardinality of a bijection,  denoted as $\gamma_{\infty}$, where $d_{\infty}(A,B)$ is reached. 
If $r_{\infty}(A,B) \leq \frac{2}{3e}$ then:
$$		||S(A)- S(B)||_1 \leq r_\infty(A,B)L_{max}\left( \frac{\log n_{\max}}{n_{\max}}-\frac{3}{2}\log \Big(\frac{3}{2}r_\infty(A,B) \Big)
			 \right).	
$$
	\end{theorem}

		\begin{figure}[!htb]
   		\begin{tabular}{c c}
 \includegraphics[height=3.5cm]{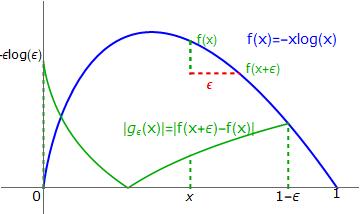}
 &
  \includegraphics[height=3.5cm]{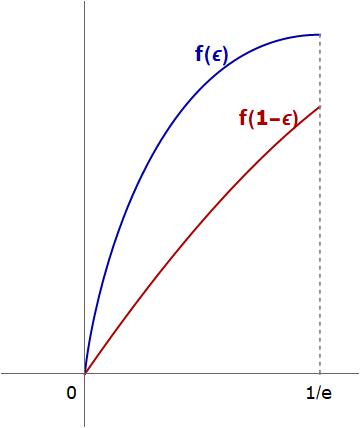}
   	\end{tabular}
 	\caption{On the left, 
 	$|g_{\epsilon}(x)|$ 
 	is pictured in green and
 	$f(x)$
 	in blue. 
  	On the right, observe that $f(\epsilon)> f(1-\epsilon)$ 
  	for  $\epsilon\in(0,1/e)$, so $|g_{\epsilon}(x)|$ attains
 	the global  maximum on $[0,1-\epsilon]$ 
 	at $x=0$.
 	}\label{figure:fig_xlogx}
      \end{figure}
     
	\begin{proof}
		Denote the expression $\frac{\ell_i^a}{L_a} \log \left( \frac{\ell_i^a}{L_a} \right)$ by $s_i^a$. Then,
		\begin{align*}
			&||S(A) - S(B)||_1  
		=\bigg|\bigg| \sum_i^{n_{\infty}} w^a_is_i^a - w^b_{\gamma_{\infty}(i)}s^b_{\gamma_{\infty}(i)} \bigg|\bigg|_1\\
		&=\bigg|\bigg| \sum_i^{n_{\infty}} \big(w^a_iw^b_ {\gamma_{\infty}(i)} + w^a_i(1-w^b_ {\gamma_{\infty}(i)})\big)s_i^a   
		-\big(w^b_ {\gamma_{\infty}(i)}w^a_i + w^b_ {\gamma_{\infty}(i)}(1-w^a_i)\big)s^b_{\gamma_{\infty}(i)} \bigg|\bigg|_1 \nn \\
		&= \bigg|\bigg| \sum_i^{n_{\infty}} w^a_iw^b_ {\gamma_{\infty}(i)}(s^a_i - s^b_ {\gamma_{\infty}(i)})
			+ w^a_i(1-w^b_ {\gamma_{\infty}(i)})s_i^a 
					-  w^b_ {\gamma_{\infty}(i)}(1-w^a_i)s^b_ {\gamma_{\infty}(i)} \bigg|\bigg|_1 \nn \\
		&\leq \sum_i^{n_{\infty}} \big|\big| w^a_iw^b_ {\gamma_{\infty}(i)}\big|\big|_1 |s^a_i - s^b_ {\gamma_{\infty}(i)}|  
				+ \big|\big| w^a_i(1-w^b_ {\gamma_{\infty}(i)})s_i^a \big|\big|_1 
					+ \big|\big| w^b_ {\gamma_{\infty}(i)}(1-w^a_i)s^b_ {\gamma_{\infty}(i)} \big|\big|_1.\nn
		\end{align*}
		Let us now compute a bound for $ \sum_{i=1}^{n_{\infty}} \big|\big| w^a_iw^b_ {\gamma_{\infty}(i)}\big|\big|_1 |s^a_i - s^b_ {\gamma_{\infty}(i)}|$. 
		Without loss of generality, assume  that $L_a\geq L_b$.	Since $\big|\big| w^a_iw^b_ {\gamma_{\infty}(i)}\big|\big|_1$ represents the length of the intersection of non-null paired intervals, we have that
		\begin{align}\label{complicado27}
	\sum_i^{n_{\infty}} || w^a_i
	w^b_ {\gamma_{\infty}(i)} ||_1 \leq L_a.
		\end{align}
      Let us compute a bound for $|s^a_i - s^b_ {\gamma_{\infty}(i)}|$. Denote  the expression $	\left|\frac{\ell_i^a}{L_a} - \frac{ \ell^b_{\gamma_{\infty}(i)}}{L_b}\right|$
      by $\epsilon$.  By Lemma \ref{lema-tecnico} and Lemma \ref{Projection-inequality},  we have that 
        \begin{align}\label{complicado3}
		\epsilon 
		&\leq  \frac{ \big|\ell^a_i - \ell_{\gamma_{\infty}(i)}^b\big|}{L_a} + \frac{ \ell^b_{\gamma_{\infty}(i)}d_1(\pi(A),\pi(B))}{L_a L_b}  \nn \\
		&\leq \frac{2\max \{ |x_i^a - x^b_ {\gamma_{\infty}(i)}|, |y_i^a - y^b_ {\gamma_{\infty}(i)}| \}}{L_a}  + \frac{ L_b d_1(\pi(A),\pi(B))}{L_a L_b} \nn \\ 
		&\leq \frac{2d_{\infty}(A,B)}{L_a}  + \frac{2 d_1(A,B)}{L_a}
		\leq \frac{2d_{\infty}(A,B)}{L_a}  + \frac{2n_{\infty} d_{\infty}(A,B)}{L_a} \nn \\
		&\leq \frac{(2+2n_{\infty}) d_{\infty}(A,B)}{L_a} \leq  \frac{3n_{\infty} d_{\infty}(A,B)}{L_a} 
		=
		\frac{3}{2}r_\infty(A,B).
		\end{align}
		By hypothesis, we have that $r_{\infty}(A,B) \leq \frac{2}{3e}$, so  $ \frac{3}{2}r_\infty(A,B)\leq \frac{1}{e}$. 
		Then, $\epsilon \leq \frac{1}{e}$.
		Now, recall that  $f(x) = -x\log x$ is continuous and concave in $[0,1]$ with $f(0)=f(1)=0$ and consider the function $g_{\epsilon}(x)=f(x + \epsilon) - f(x)$ defined for any $x\in[0,1-\epsilon]$. 	As $f$ is concave and differentiable in $(0,1)$, $f'$ is decreasing, $g'_{\epsilon}$ is negative and so $g_{\epsilon}$ is decreasing  monotone.
		Therefore the maximum of its absolute value is attained at one of the extreme points of the interval $[0,1-\epsilon]$. Observe that since  $\epsilon < \frac{1}{e}$, 
		this maximum is reached at
		$x=0$ (see Figure~\ref{figure:fig_xlogx}):
		
		\begin{align}\label{complicado2}\begin{array}{ll}
		    |g_{\epsilon}(x)|&=|f(x + \epsilon) - f(x) |\leq \max\{ |f(\epsilon) - f(0)|, |f(1) - f(1-\epsilon)| \} \\
		    & = \max\{ -\epsilon\log(\epsilon), -(1-\epsilon)\log(1-\epsilon)\} = -\epsilon \log(\epsilon).
		    \end{array}
		\end{align}
 Now, using (\ref{complicado2}), we obtain that:
      
      	 \begin{align}\label{complicado1}
		|s^a_i - s^b_ {\gamma_{\infty}(i)}| =\left| \frac{\ell_i^a}{L_a}\log\left(\frac{\ell_i^a}{L_a}\right) - \frac{\ell^b_ {\gamma_{\infty}(i)}}{L_b}\log\left(\frac{\ell^b_ {\gamma_{\infty}(i)}}{L_b}\right) \right| 
		\leq   -\epsilon \log(\epsilon).
		\end{align}
Due to (\ref{complicado3}) and $f(x)=-x\log x$ being increasing in $[0,\frac{1}{e}]$, 
we obtain from (\ref{complicado1}) that:
		 \begin{align}\label{complicado28}
		|s^a_i - s^b_ {\gamma_{\infty}(i)}| \leq   
	 -\dfrac{3}{2}r_\infty(A,B) \log \Big( \dfrac{3}{2}r_\infty(A,B) \Big).
	 \end{align}
		Finally, by (\ref{complicado27}) and (\ref{complicado28}), we have:
		\begin{align}
		\sum_i^{n_{\infty}} \big|\big| w^a_iw^b_ {\gamma_{\infty}(i)}\big|\big|_1 |s^a_i - s^b_ {\gamma_{\infty}(i)}| \leq -\frac{3}{2}
		L_a r_\infty(A,B) \log \left( \frac{3}{2}r_\infty(A,B) \right). \label{addend1}
		\end{align}  	        
		Now, let us compute a bound for $$\sum_i^{n_{\infty}} \big|\big| w^a_i(1-w^b_ {\gamma_{\infty}(i)})s_i^a \big|\big|_1 + \big|\big| w^b_ {\gamma_{\infty}(i)}(1-w^a_i)s^b_ {\gamma_{\infty}(i)} \big|\big|_1.$$ Consider the function $w^b_ {\gamma_{\infty}(i)}(1-w^a_i)$. 	Its integral gives  the 
	``period of time'' in which the 
		${\gamma_{\infty}(i)}$-th interval of $B$, $[x^b_ {\gamma_{\infty}(i)}, y^b_ {\gamma_{\infty}(i)})$, is ``alive'' and the $i$-th interval of $A$, $[x_i^a, y_i^a)$, is not. 
		This might happen in both the initial and the end of the intervals.
		Therefore, if  $\epsilon_i = \max \{ |x_i^a - x^b_ {\gamma_{\infty}(i)}|, |y_i^a - y^b_ {\gamma_{\infty}(i)}| \}$ then:
		\begin{equation*}
		\int_\mathbb{R} w^b_ {\gamma_{\infty}(i)}(t)(1-w^a_i(t))\leq 2\epsilon_i.
		\end{equation*} 
		We also have that, in the case where $\epsilon_i < \int_\mathbb{R} w^b_ {\gamma_{\infty}(i)}(t)(1-w^a_i(t)) dt$,
		the period of time  where the $i$-th interval of $A$ is alive and the one of $B$ is not, is null. Therefore
		\begin{equation*}
		\epsilon_i < \int_\mathbb{R} w^b_ {\gamma_{\infty}(i)}(t)(1-w^a_i(t)) dt  \then \int_\mathbb{R} w^a_i(t)(1-w^b_ {\gamma_{\infty}(i)}(t)) dt = 0
		\end{equation*} 
		and vice-versa. Using both previous statements and that $\sum_{i=1}^{n_a} s_i^a = E(A)$ we can deduce:
		\begin{align*}
		&\sum_i^{n_{\infty}}s_i^a \bigg|\bigg| w^a_i(1-w^b_ {\gamma_{\infty}(i)}) \bigg|\bigg|_1+ s^b_ {\gamma_{\infty}(i)} \bigg|\bigg| w^b_ {\gamma_{\infty}(i)}(1-w^a_i) \bigg|\bigg|_1 \nn \\
		&\leq \sum_i^{n_{\infty}} s_i^a\int_{\mathbb{R}} w^a_i(1-w^b_ {\gamma_{\infty}(i)}) + s^b_ {\gamma_{\infty}(i)} \int_{\mathbb{R}} w^b_ {\gamma_{\infty}(i)}(1-w^a_i) 
		\nn \\ 
		&
		\leq \max\bigg\{ \sum_i^{n_{\infty}}  \epsilon_i(s_i^a + s^b_ {\gamma_{\infty}(i)}), \sum_i^{n_{\infty}} 2\epsilon_i s_i^a, \sum_i^{n_{\infty}} 2\epsilon_i s^b_ {\gamma_{\infty}(i)}\bigg\}  \nn \\
		&
		\leq \max\bigg\{ \max_{i}\{\epsilon_i\} \left(\sum_i^{n_{\infty}}  s_i^a + \sum_i^{n_{\infty}} s^b_ {\gamma_{\infty}(i)} \right), 2\max_{i}\{\epsilon_i\} \sum_i^{n_{\infty}}  s_i^a, 2\max_{i}\{\epsilon_i\}\sum_i^{n_{\infty}}  s^b_ {\gamma_{\infty}(i)} \bigg\}  \nn \\
				&
		= \max\bigg\{ \max_{i}\{\epsilon_i\} \left(\sum_i^{n_a}  s_i^a + \sum_i^{n_b} s^b_i \right), 2\max_{i}\{\epsilon_i\} \sum_i^{n_a}  s_i^a, 2\max_{i}\{\epsilon_i\}\sum_i^{n_b}  s^b_{i}\bigg\}  \nn \\
		& \leq  d_\infty(A,B) \max \big\{E(A) + E(B), 2E(A), 2E(B)\big\} 
		\leq r_{\infty}(A,B)\frac{L_{a}}{n_{\max}} \log( n_{\max}).
		\end{align*}
From this last equation and \eqref{addend1}, we obtain the desired result. 		
	\end{proof}

Notice that the  ES-function is based on persistent entropy whereas the Betti  curve  consists of  counting the number of ``alive'' intervals. Both functions (the ES-function and the Betti  curve) are continuous with respect to the bottleneck distance if the maximum number of intervals is fixed. Nevertheless, the ES-function is expected to perform better than the Betti curve in a noisy context since persistent entropy is stable while counting the number of intervals is not, even if it is continuous. 

\subsection{Normalized entropy summary function (NES-function)}
	
	One of the main aims of persistent homology is to represent the shape of the input data. In some applications, like image analysis or material science (see \cite{Material-Review} for a review), it may be important to detect some repetitive pattern independently of the size of the input dataset. 
	A possible tool to do this is a normalized version of the summary function, in order to try to capture the shape of the space and not the size.
	
	\begin{definition}[NES-function]
		The normalized entropy summary function 
		(NES-function)
		of 
 a persistence barcode $A = \{ [x_i^a, y_i^a] \}_{1\leq i\leq n_a}$ in $\mathcal{B}_F$ is defined as:
		\begin{equation*}
		NES(A)[t] = \frac{S(A)[t]}{||S(A)||_1}.\nn
		\end{equation*}
	\end{definition}
	
 Like the ES-function, this function is also stable.

	\begin{theorem}[Stability of the NES-function]
		Under the same hypothesis as in Theorem \ref{Es-Stability}, we have that:
\begin{align*}
		||NES(A) - & NES(B)||_1 
		\leq \dfrac{r_\infty(A,B)L_{max}\left( \frac{\log n_{\max}}{n_{\max}}-\frac{3}{2}\log \left[\frac{3}{2}r_\infty(A,B) \right]
			 \right)}{min\left\{||S(A)||_1, ||S(B)||_1\right\}}.
		\end{align*}
	\end{theorem}
	\begin{proof} First, observe that
	\begin{align*}
	&	\left| \left| \dfrac{S(A)}{||S(A)||_1} - \dfrac{S(B)}{||S(B)||_1} \right|\right|_1 = \dfrac{\big|\big| ||S(B)||_1 S(A) - ||S(A)||_1 S(B) \big|\big|_1}{||S(A)||_1||S(B)||_1} \\
		&\leq \dfrac{\max\{||S(A)||_1 ||S(B)||_1 \}(||S(A) - S(B)||_1)}{||S(A)||_1||S(B)||_1} = \dfrac{||S(A) - S(B)||_1}{\min\{ ||S(A)||_1,||S(B)||_1\}}.
	\end{align*}
	Apply Theorem~\ref{Es-Stability} to bound $||S(A) - S(B)||_1$ obtaining the desired result.
	\end{proof}
	
	\section{Experimentation}
	
Our summary functions have been recently applied to real-world 
data 
	in \cite{skin} and \cite{lawson}. The authors of these papers explicitly mentioned the stability of these functions, presented here and in the arXiv version \cite{preprint}, to guarantee the robustness of their methods. In \cite{skin} an algorithm is developed for segmenting and classifying different types of skin lesions in a given skin image. 
The database used is the International Skin Imaging Collaboration (ISIC) dataset\footnote{https://www.isic-archive.com} consisting of 10015 skin lesson images.
 In the classification part, the authors
first computed the persistence diagram 
for each channel of the image in different color space. Then, 
  topological features including our ES-function (referred as ``persistent entropy curves'') and Betti curves were calculated from persistence diagrams obtained previously.
  Finally,
the authors used multi-class support vector machine with a ``one-against-one''
strategy.
The best 3 scores they had on the validation set was 65.6\%,
66\%, and 67.2\% depending on the color space used. 
Besides, in \cite{lawson}, the authors developed a general framework called ``persistence curves'' for vectorizing persistence diagrams  inspired in our ES-function, referred in  \cite{lawson} as ``life entropy'', denoted by $le$, and used 
for texture classification on four different texture datasets. They showed that a combination of different flavors of persistence curves (including our ES-function) produces the best result, showing that they are complementary to each other.

This section is devoted to experiments. In Subsection \ref{USC-SIPI experiment} we  
study
how the NES-function and similar vectorizations such as the
Betti curves and the persistence silhouettes
may benefits from each other in machine learning tasks.
The dataset considered in our experiment consists of miscellaneous real-world images taken from the USC-SIPI Image Database\footnote{http://sipi.usc.edu/database/database.php?volume=misc}. The machine learning method used in our experiment is the random forest technique \cite{RF}.
Later, in Subsection \ref{subsec:3.2}, we will also add persistence images to our experiment to see how different vectorization methods perform depending on the nature of the data sets. In particular, we will consider the Flickr Material Database (FMD)\footnote{https://people.csail.mit.edu/celiu/CVPR2010/FMD/}.
	The whole experiment has been developed in Python and R. Scripts and notebooks can be downloaded from here\footnote{https://github.com/Cimagroup/New-Summary-Function-For-TDA}. Libraries used are scikit-TDA\footnote{https://scikit-tda.org/} (in particular Persim\footnote{https://persim.scikit-tda.org/} and Ripser\footnote{https://ripser.scikit-tda.org/}), scikit-image\footnote{https://scikit-image.org/}, scikit-learn\footnote{https://scikit-learn.org/} and R-TDA-package\footnote{https://cran.r-project.org/web/packages/TDA/index.html}.
	
\subsection{USC-SIPI dataset experiment}\label{USC-SIPI experiment}

\begin{figure}[ht]
	\centering
	\includegraphics[width = .75\textwidth]{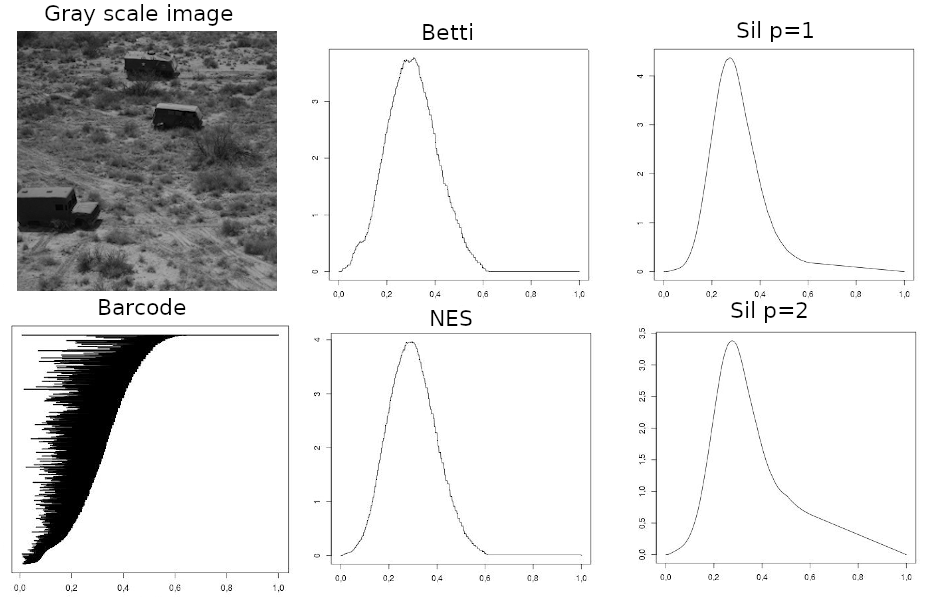}
	\caption{Left: a gray scale image from the USC-SIPI Image Database and its associated $0$-th persistence barcode used to compute the summary functions. Center: the Betti curve and the NES-function. Right: The persistence silhouettes for $p=1$ and $p=2$. Observe that, in this example, outputs are extremely similar.}
	\label{fig:29}
\end{figure}
\begin{figure}[ht]
	\centering
	\includegraphics[width = .75\textwidth]{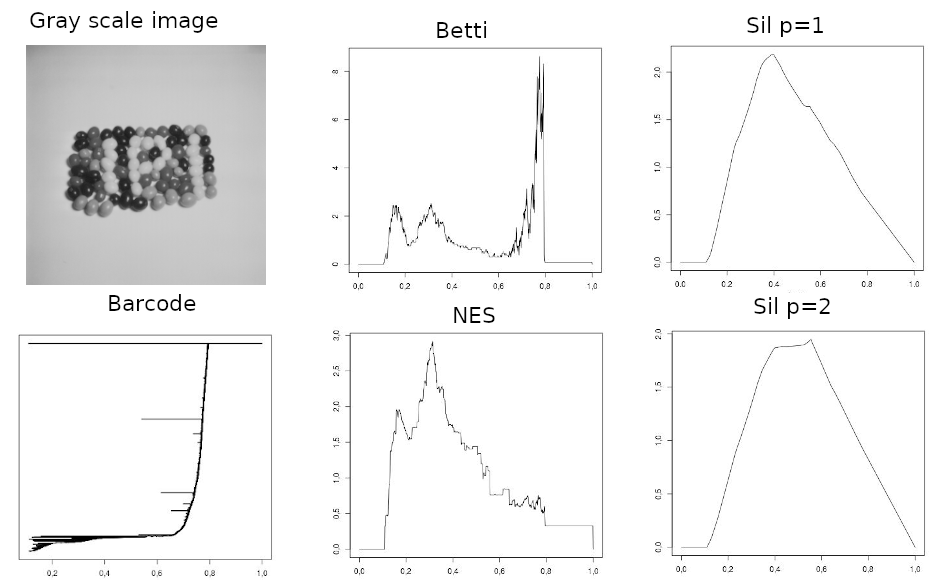}
	\caption{An image producing very different output functions}
	\label{fig:06}
\end{figure}

\begin{table}[th]
	\centering
	\caption{Average of the
		$L_1$-distance of the selected summary functions on
		the $0$-th persistence barcodes associated to the images from USC-SIPI Image Database and 
		a noisy version of them. Note that the maximum possible value  is 2 since all summary functions used have been normalized.}
	\label{table:robustness}
	\medskip
	{\scriptsize
	\begin{tabular}{|l|l|l|l|l|}
		\hline
		& Betti curve    & NES-function       & Silhouette $p=1$    & Silouette $p=2$    \\ \hline
		Gauss   & 0.198 & 0.114 & 0.039 & 0.024 \\ \hline
		Poisson & 0.245 & 0.200 & 0.156 & 0.086 \\ \hline
		S\&P    & 0.140 & 0.296 & 0.424 & 0.459 \\ \hline
	\end{tabular}
	}
\end{table}

\begin{table}[th]
	\centering
	\caption{Images in the database are quite different to each other so these functions are expected to discriminate them. This table shows the average of the 
		$L_1$-distance of the considered summary functions  on all the images of the given dataset.
		Note that the maximum possible value  is 2 since all summary functions used have been normalized.}
	\label{table:differences}
		\medskip
		{\scriptsize
	\begin{tabular}{|l|l|l|l|}
		\hline
		Betti  curve   & NES-function       & Silhouette $p=1$   & Silhouette $p=2$   \\ \hline
		0.913 & 0.672 & 0.500 & 0.354 \\ \hline \end{tabular}}
\end{table}

We have followed this procedure:
\begin{itemize}
	\item[] Step 1.Transform the images to gray scale.
	\\
	Step 2. Add  Gaussian, Poisson and salt-and-pepper noise to the data.
	\\
	Step 3. Compute persistence diagrams using lower-start filtration.
	\\
	Step 4. Summarize the diagrams using some vectorization method: the Betti curve, the NES-function or the persistence silhouettes.
\end{itemize}
Let us compare the results obtained. First, the computational time was very similar to obtain
all of them so we have omitted it from the analysis. A first conclusion is that for some images like the one in Figure~\ref{fig:29}, the outputs have  extremely similar shapes while in others, like the one in Figure~\ref{fig:06}, the outputs have completely different shapes, showing that these functions may provide complementary information regarding different aspects of the same image. Robustness results have been obtained computing the $L_1$-distance
of each summary function on clean and noisy images,
see Table~\ref{table:robustness}. Observe that the results obtained for the NES-function is always between those obtained  for the Betti curve and the persistence silhouettes. The Betti curve only performs better for salt-and-pepper noise. This fact is expected since persistence diagrams together with  the bottleneck distance are unstable to salt-and-pepper noise when they are calculated using the lower-star filtration. It tells us that the NES-function and the persistence silhouettes are more robust to stastistical noise than the Betti curve. However, the Betti curve is more robust to impulsive noise. 

We have also pairwise compared  all the clean images for each summary function using the 
$L_1$-distance and, as expected, since all images in the database are different in nature, this value is  high, see Table~\ref{table:differences}.

We can conclude that the NES-function is usually more discriminative than the persistence silhouettes but less than
the Betti curve. Nevertheless, these functions may provide information regarding different aspects of the same image and therefore may complement each other as it is shown in the next subsection.

\subsection{FMD dataset experiment}\label{subsec:3.2}

First, we have transformed the given color images to gray scale images and computed their lower-start filtration. In a first part of this experiment, we have computed the Betti curve, the NES-function and the persistence silhouette for $p = 1$. Materials in the FMD database are classified in ten categories: 0-foliage, 1-glass, 2-leather, 3-metal, 4-paper, 5-plastic, 6-stone, 7-water, 8-wood and 9-fabric. Then, we have applied the random forest technique to classify the images using the output of the summary functions. Results are shown in Table~\ref{table:bsn}. Note that the one that performs best is the combination of the Betti curve and the NES function. Besides, the persistence silhouettes also improve their performance when combined with the NES-function.

\begin{table}[h!]
	\centering
	\caption{Accuracy of the classification of the FMD database using the random forest technique  and the summary functions selected. Note that both, the Betti curve and the persistence silhouettes improve when combined with the NES-function.}
	\label{table:bsn}
		\medskip
		{\scriptsize
\begin{tabular}{|l|l|l|l|l|l|l|}
	\hline
	Betti                       & Silhouette   $p=1$              & Nes                       & B + S                      & B + N                              & S + N                     & B + S + N                 \\ \hline
	\multicolumn{1}{|c|}{0.285} & \multicolumn{1}{c|}{0.185} & \multicolumn{1}{c|}{0.24} & \multicolumn{1}{c|}{0.285} & \multicolumn{1}{c|}{\textbf{0.29}} & \multicolumn{1}{c|}{0.25} & \multicolumn{1}{c|}{0.265} \\ \hline
\end{tabular}
}
\end{table}

In a second part of this experiment, we want to illustrate that, depending on the data, some methods may perform better than others. This is the reason why this time we do not combine functions. The idea is to try to distinguish categories in the database when compared pairwise. We have performed 45 tests, one for each pair of materials. We have added another popular vectorization method: persistence images (with $20\times 20$ pixels). Results are shown in Table~\ref{table:pw} where best marks have been highlighted.

\begin{table}[h!]
	\centering
	\caption{We perform a classification task for each pair of categories in FMD. It can be checked that persistence images and the Betti curve usually performs better but the persistence silhouettes and the NES-function may outperform the Betti curve and the persistence images in some cases, concluding that the information provided by each function are complementary and depend on the data.}
	\label{table:pw}
		\medskip
		{\scriptsize
	\begin{tabular}{l|lllllllll|}
		\cline{2-10}
		& \multicolumn{1}{l|}{0vs1} & \multicolumn{1}{l|}{0vs2} & \multicolumn{1}{l|}{0vs3} & \multicolumn{1}{l|}{0vs4} & \multicolumn{1}{l|}{0vs5} & \multicolumn{1}{l|}{0vs6} & \multicolumn{1}{l|}{0vs7} & \multicolumn{1}{l|}{0vs8} & \multicolumn{1}{l|}{0vs9} \\ \hline
		\multicolumn{1}{|l|}{B}  & 0.6                       & 0.675                     & 0.7                       & 0.6                       & 0.575                     & 0.725                     & 0.675                     & 0.8                       & 0.775                     \\ \cline{1-1}
		\multicolumn{1}{|l|}{S}  & \textbf{0.625}             & 0.525                       & 0.55                     & 0.45                      & 0.55                      & 0.75                     & 0.575                     & 0.775                      & 0.675                       \\ \cline{1-1}
		\multicolumn{1}{|l|}{N}  & 0.5                       & 0.6                       & 0.7                       & 0.5                       & 0.45                      & 0.725                     & 0.7                       & 0.8                       & 0.775          \\ \cline{1-1}
		\multicolumn{1}{|l|}{PI} & 0.55                      & \textbf{0.725}            & \textbf{0.8}              & \textbf{0.675}            & \textbf{0.6}              & \textbf{0.775}            & \textbf{0.875}            & \textbf{0.85}             & \textbf{0.825}            \\ \hline
		& \multicolumn{1}{l|}{1vs2} & \multicolumn{1}{l|}{1vs3} & \multicolumn{1}{l|}{1vs4} & \multicolumn{1}{l|}{1vs5} & \multicolumn{1}{l|}{1vs6} & \multicolumn{1}{l|}{1vs7} & \multicolumn{1}{l|}{1vs8} & \multicolumn{1}{l|}{1vs9} & \multicolumn{1}{l|}{2vs3} \\ \hline
		\multicolumn{1}{|l|}{B}  & 0.8                       & 0.55                      & 0.75                      & \textbf{0.75}             & 0.725                     & 0.575                     & 0.65                      & 0.575                     & 0.475                     \\ \cline{1-1}
		\multicolumn{1}{|l|}{S}  & 0.825                     & 0.65                     & \textbf{0.8}                       & 0.6                     & 0.75                       & 0.475                       & \textbf{0.775}                     & 0.6                     & 0.45                     \\ \cline{1-1}
		\multicolumn{1}{|l|}{N}  & \textbf{0.85}             & 0.6                       & 0.775            & \textbf{0.75}             & 0.775                     & \textbf{0.625}            & 0.675            & \textbf{0.675}            & 0.5                       \\ \cline{1-1}
		\multicolumn{1}{|l|}{PI} & \textbf{0.85}             & \textbf{0.775}            & 0.775            & 0.675                     & \textbf{0.875}            & 0.575                     & 0.625                     & \textbf{0.675}            & \textbf{0.6}              \\ \hline
		& \multicolumn{1}{l|}{2vs4} & \multicolumn{1}{l|}{2vs5} & \multicolumn{1}{l|}{2vs6} & \multicolumn{1}{l|}{2vs7} & \multicolumn{1}{l|}{2vs8} & \multicolumn{1}{l|}{2vs9} & \multicolumn{1}{l|}{3vs4} & \multicolumn{1}{l|}{3vs5} & \multicolumn{1}{l|}{3vs6} \\ \hline
		\multicolumn{1}{|l|}{B}  & \textbf{0.825}            & 0.7                       & \textbf{0.7}              & 0.525                     & \textbf{0.65}             & 0.65                      & \textbf{0.725}            & \textbf{0.925}            & \textbf{0.825}            \\ \cline{1-1}
		\multicolumn{1}{|l|}{S}  & 0.725                       & 0.575                     & 0.5                    & 0.5            & 0.6                       & 0.55                   & 0.675                      & 0.65                      & 0.6                      \\ \cline{1-1}
		\multicolumn{1}{|l|}{N}  & 0.725                     & 0.65                      & 0.55                      & 0.475                     & \textbf{0.65}             & \textbf{0.75}             & 0.7                       & 0.85                      & 0.625                     \\ \cline{1-1}
		\multicolumn{1}{|l|}{PI} & 0.8                       & \textbf{0.8}              & 0.625                     & \textbf{0.625}            & 0.575                     & 0.7                       & 0.65                      & 0.775                     & 0.775                     \\ \hline
		& \multicolumn{1}{l|}{3vs7} & \multicolumn{1}{l|}{3vs8} & \multicolumn{1}{l|}{3vs9} & \multicolumn{1}{l|}{4vs5} & \multicolumn{1}{l|}{4vs6} & \multicolumn{1}{l|}{4vs7} & \multicolumn{1}{l|}{4vs8} & \multicolumn{1}{l|}{4vs9} & \multicolumn{1}{l|}{5vs6} \\ \hline
		\multicolumn{1}{|l|}{B}  & \textbf{0.7}              & 0.625                     & 0.725                     & 0.55                      & \textbf{0.825}            & 0.725                     & 0.625                     & 0.55                      & 0.65                      \\ \cline{1-1}
		\multicolumn{1}{|l|}{S}  & 0.625                     & 0.6                     & 0.675            & \textbf{0.625}                     & 0.7                     & 0.65                    & 0.45                      & 0.6                     & 0.6                     \\ \cline{1-1}
		\multicolumn{1}{|l|}{N}  & \textbf{0.7}              & 0.65                      & 0.725                     & 0.55                      & 0.75                      & 0.725                     & 0.6                       & \textbf{0.65}             & \textbf{0.725}            \\ \cline{1-1}
		\multicolumn{1}{|l|}{PI} & 0.6                       & \textbf{0.725}            & 0.675                     & \textbf{0.625}            & 0.775                     & \textbf{0.775}            & \textbf{0.65}             & 0.625                     & 0.6                       \\ \hline
		& \multicolumn{1}{l|}{5vs7} & \multicolumn{1}{l|}{5vs8} & \multicolumn{1}{l|}{5vs9} & \multicolumn{1}{l|}{6vs7} & \multicolumn{1}{l|}{6vs8} & \multicolumn{1}{l|}{6vs9} & \multicolumn{1}{l|}{7vs8} & \multicolumn{1}{l|}{7vs9} & \multicolumn{1}{l|}{8vs9} \\ \hline
		\multicolumn{1}{|l|}{B}  & \textbf{0.85}             & 0.7                       & \textbf{0.8}              & \textbf{0.9}              & \textbf{0.825}            & \textbf{0.8}              & \textbf{0.875}            & 0.65                      & \textbf{0.9}              \\ \cline{1-1}
		\multicolumn{1}{|l|}{S}  & 0.7                     & 0.7                     & 0.625                      & 0.625                      & 0.675                     & 0.675                     & 0.7                     & 0.7                       & 0.575                      \\ \cline{1-1}
		\multicolumn{1}{|l|}{N}  & 0.7                       & 0.675                     & 0.675                     & 0.775                     & 0.775                     & 0.725                     & 0.7                       & 0.7                       & 0.625                     \\ \cline{1-1}
		\multicolumn{1}{|l|}{PI} & 0.825                     & \textbf{0.725}            & 0.55                      & 0.7                       & 0.775                     & 0.675                     & 0.725                     & \textbf{0.725}            & 0.625                     \\ \hline \cline{1-1}
	\end{tabular}}
\end{table}

	
	\section{Conclusions and future work}
In this paper, the stability of persistent entropy is provided justifying its application as an useful statistic in topological data analysis. What is more, persistent entropy has been used to define an stable summary function, the ES-function, and its normalised version, the NES-function. We have shown that, in general, they perform better than the Betti  curve in noisy context and 
that
they can be useful for machine learning tasks.

Several types of persistence curves inspired in our persistent entropy summaries were also defined by Y.M Chung and A. Lawson in \cite{lawson}. They also corroborate the idea that such persistence functions are somehow complementary and combined provide a classification performance comparable to the state of the art. Together with Y.M. Chung and A. Lawson we plan to apply  our summary functions to other higher dimensional  non-image dataset such as, for example, the TOSCA dataset of 3D meshes\footnote{http://tosca.cs.technion.ac.il/book/resources\_data.html}. Besides, as a future work,
it would be interesting to 
deeply
compare these and other summary functions with other topological vectorization methods
existing in the literature.
\\
\noindent{\bf Acknowledgments:} We would like to thank the reviewers for their very valuable comments and suggestions. The third author has been partially funded by VI-PPITUS (University of Seville).

	\bibliographystyle{ieeetr}
	\bibliography{references}
	
\end{document}